\documentclass[oneside, 11pt, reqno]{amsart}
\usepackage[left=1in, right=1in, top=1in, bottom=1in]{geometry}

\usepackage[usenames,dvipsnames]{xcolor}
\usepackage{amsmath}
\DeclareMathOperator*{\argmax}{argmax}
\usepackage{graphicx,enumitem}
\usepackage{mathtools}

\usepackage{amssymb} 
\usepackage{soul}  

\usepackage[margin=0.7cm]{caption}
\usepackage{multirow}

\DeclareMathOperator\sgn{sgn}

\providecommand{\R}{} \renewcommand{\R}{{\mathbb R}}
 
\newcommand{\N}{{\mathbb N}}
\newcommand{\PP}{{\mathbb P}}
\newcommand{\EE}{{\mathbb E}}
\newcommand{\FF}{{\mathcal F}}

\newcommand{\Zitkovic}[1]{{\v Z}itkovi\'c}
\newcommand{\Sirbu}[1]{S\^\i rbu}

\newcommand\tV{{\tilde{V}}}
\newcommand\ue{\underline{\epsilon}}
\newcommand\be{\bar{\epsilon}}
\newcommand\uy{\underline{y}}
\newcommand\oy{\bar{y}}

\numberwithin{equation}{section}
\theoremstyle{plain}                
\newtheorem{theorem}{Theorem}[section]
\newtheorem{lemma}[theorem]{Lemma}
\newtheorem{proposition}[theorem]{Proposition}
\newtheorem{corollary}[theorem]{Corollary}

\theoremstyle{definition}           

\newtheorem{assumption}[theorem]{Assumption}
\newtheorem{notation}[theorem]{Notation}
\theoremstyle{remark}
\newtheorem{remark}[theorem]{Remark}

\newcommand{\thmref}[1]{Theorem~\ref{#1}}
\newcommand{\proref}[1]{Proposition~\ref{#1}}

\newcommand{\lemref}[1]{Lemma~\ref{#1}}

\newcommand{\remref}[1]{Remark~\ref{#1}}
\newcommand{\assref}[1]{Assumption~\ref{#1}}
\usepackage{graphicx}

\begin{document}

\author{Jin Hyuk Choi, Tae Ung Gang}
\thanks{Department of mathematics, Ulsan National Institute of Science and Technology; jchoi@unist.ac.kr, gangtaeung@unist.ac.kr}

\title[Optimal investment with search frictions and transaction costs]{Optimal investment in illiquid market \\with search frictions and transaction costs}

\maketitle

\begin{abstract}
We consider an optimal investment problem to maximize expected utility of the terminal wealth, in an illiquid market with search frictions and transaction costs. In the market model, an investor's attempt of transaction is successful only at arrival times of a Poisson process, and the investor pays proportional transaction costs when the transaction is successful. 
We characterize the no-trade region describing the optimal trading strategy. 
We provide asymptotic expansions of the boundaries of the no-trade region and the value function, for small transaction costs. The asymptotic analysis implies that the effects of the transaction costs are more pronounced 
in the market with less search frictions.
\end{abstract}

\bigskip


\bigskip

\section{Introduction}

Understanding the effects of liquidity on optimal investment is one of the main topics in mathematical finance and financial economics. According to \cite{yakov2005liquidity}, stated simply, liquidity is the ease of trading a security. Various sources of illiquidity include exogenous transaction costs, search frictions such as the difficulty of locating a counterparty with whom to trade, and price impacts due to private information.\footnote{For instance, \cite{Kyl85, bac92, asym, CHOI201922} study how asymmetric information effects price impact and optimal trading strategy in equilibrium. \cite{almgren2001optimal, gatheral2011optimal, predoiu2011optimal, robert2012measuring} consider optimal order execution  problems with exogenously given price impacts.} 
 This paper studies an optimal investment problem in the market model with two different types of illiquidity: search frictions and transaction costs.

Under the assumption of the perfect liquidity,\footnote{Here, perfect liquidity assumption means that assets can be traded in any quantity and at any moment in time, without any transaction costs.} the pioneering papers \cite{Mer69, Mer71} formulate the optimal investment problem (so-called Merton's portfolio problem) with a geometric Brownian motion and a CRRA (constant relative risk aversion) investor, and show that the optimal solution is to keep the constant fraction of wealth invested in the risky asset.  With the same assumption of perfect liquidity, more general stochastic processes and utility functions have been considered to obtain more general characterizations of the optimal investment strategies (e.g., \cite{KLS87, KLSX91, KraSch99, KarZit03, Kra04}). 

In the optimal investment problems, the assumption of perfect liquidity can be relaxed by considering search frictions in the market. As Table 1 in \cite{Andrew14} shows, many financial assets are illiquid in the sense that it is difficult to find a counterparty who is willing to trade. One way to incorporate this type of illiquidity, search frictions, into the optimal investment problem is to impose some restrictions on trade times.
In the classical Merton framework, \cite{Rogers:2001aa} considers an investor who is allowed to change portfolio only at times which are multiple of a constant $h>0$, and \cite{Rogers2002, matsumoto2006, Andrew14} assume that an illiquid asset can only be traded on the arrival of a randomly occurring trading opportunity that is represented by jump times of a Poisson process. \cite{2008Pham, 2011Pham} consider an optimal investment/consumption problem under the assumption that the asset price is observed only at the random trade times. \cite{2014Pham} complicates the model by using random intensity of trade times, regime-switching, and liquidity shocks. In \cite{Dai2016}, a risky asset can be traded only on deterministic time intervals and the investor pays proportional transaction costs.

Transaction costs (such as order processing fees or transaction taxes) are another source of market illiquidity, and the optimal investment problems with transaction costs have been extensively studied in mathematical finance community. \cite{MagCon76, DavNor90, ShrSon94} study the model in \cite{Mer71} with the assumption that proportional transaction costs are levied on each transaction, and show (with different level of mathematical rigorosity) that it is optimal to keep the fraction of wealth invested in the risky asset in an interval so called {\it no-trade region}. The boundaries of the no-trade region are characterized in terms of the free-boundaries determined by the HJB (Hamilton-Jacobi-Bellman) equation of the control problem. The models with transaction costs and multiple risky assets have been studied (e.g., \cite{Aki96, Liu04, Mut06, CheDai13} for costs on all assets and \cite{Dai11, Guasoni, Hobson16, Choi2020} for costs on only one assets) to characterize the value function or no-trade region. 
More general stochastic processes have been also considered in the framework of optimal investment with transaction costs (e.g., \cite{czichowsky2014transaction, czichowsky2016duality, bayraktar2020extended}).

In this paper, we merge the aforementioned frameworks and analyze an optimal investment problem in a market model with both search frictions and transaction costs. We consider the classical Merton's portfolio  problem with a log-utility investor, whose goal is to maximize the expected utility of wealth at the terminal time $T>0$. We assume that the investor's attempt of trading is successful only when a Poisson process with intensity $\lambda$ jumps (as in \cite{Rogers2002, matsumoto2006}), and the investor needs to pay proportional transaction costs for successful trading (as in \cite{DavNor90, ShrSon94}). We show that there exists a unique classical solution of the HJB equation and provide the standard verification argument. As in the aforementioned models for transaction costs, the optimal trading strategy is characterized by a no-trade region: if the investor can trade at time $t\in [0,T)$, then there are two constants $0\leq \uy(t)\leq \oy(t)\leq 1$ such that the investor should minimally trade to keep the fraction of wealth invested in the risky asset inside of the interval $[\uy(t),\oy(t)]$. Strict concavity of the value function uniquely determines the boundaries $\uy(t)$ and $\oy(t)$  of the no-trade region.

We provide asymptotic expansions of the value function and the no-trade boundaries for small transaction costs. For the transaction cost parameter $\epsilon$, in the similar model set up without the search frictions, it is well known that the width of the no-trade region is the order of $\epsilon^{\frac{1}{3}}$ and the decrease of the value function is the order of  $\epsilon^{\frac{2}{3}}$ (e.g., see \cite{ShrSon94, JanShr04, GerMuhSch11, Cho11, BicShr11, Pos15, Guasoni, Choi2020}). In our model with the search frictions, we show that both the width of the no-trade region and the decrease of the value function are the order of $\epsilon$, not $\epsilon^{\frac{1}{3}}$ nor $\epsilon^{\frac{2}{3}}$. The search frictions in our model induce the discrete trading times, and as in the discrete time portfolio selection problems (e.g., see \cite{bayer2014utility, quek2017portfolio}), we have the $\epsilon$ order in the asymptotics.\footnote{In contrast to \cite{bayer2014utility, quek2017portfolio}, the trading times are random in our model. } 

The previous paragraph suggests that the illiquidity due to the search fractions makes the effect of the transaction costs less severe. To clarify this effect more, we differentiate the coefficients of the first-order terms with respect to the search friction parameter $\lambda$, and examine the signs of these quantities as $\lambda\to \infty$. This analysis implies that the effects of the transaction costs are more pronounced (more widening effect of the no-trade region and more diminishing effect of the value function) in the market with less search frictions.\footnote{To be more specific, let $0<\lambda_1<\lambda_2$ and consider two markets $\mathcal{M}_1$ and $\mathcal{M}_2$ with search friction parameters $\lambda_1$ and $\lambda_2$, respectively. In words, $\mathcal{M}_2$ has less search frictions than $\mathcal{M}_1$. Our result implies that if we increase the transaction costs, the widening speed of the no-trade region in $\mathcal{M}_2$ is faster than that in $\mathcal{M}_1$. Similarly, if we increase the transaction costs, the diminishing speed of the optimal value in $\mathcal{M}_2$ is faster than that in $\mathcal{M}_1$.
}

Our modeling assumption of the Poisson arrivals of the successful trading times has a similar taste to the assumptions for the liquidity provision in some models on limit order markets.    
\cite{foucault2005limit, rocsu2009dynamic} derive equilibrium order placement strategies in a limit order market model, where patient/impatient agents arrive at the market according to a Poisson process.
\cite{guilbaud2013optimal} solves an optimal market-making model with execution/inventory risks, where the limit order book is modeled by a Markov chain that jumps according to a stochastic clock described by a Poisson process.
\cite{kuhn2010optimal} studies Merton's portfolio problem where the investor has the choice between market orders (immediate transactions) and limit orders (Poisson arrivals of execution times). In the model, trading with limit orders has a flavor of negative proportional transaction costs. 

 The remainder of the paper is organized as follows. Section 2 describes the model. In Section 3, we provide the verification argument and the strict concavity of the value function. In Section 4, we characterize the optimal trading strategy in terms of the no-trade region, and present some properties of the no-trade region. In Section 5, we provide asymptotic analysis for small transaction costs. 
Section 6 summarizes this paper. 
 Proof of the technical lemmas can be found in Appendix.


\section{The Model}

Let $(\Omega, \FF, (\FF_t)_{t\geq0}, \PP)$ be a filtered probability space satisfying the usual conditions. Under the filtration, we assume that $(B_t)_{t \geq 0}$ is a standard Brownian motion and $(P_t)_{t\geq 0}$ is a Poisson process with intensity $\lambda>0$. We assume that $B$ and $P$ are independent processes.

We consider a market that has two different types of illiquidity: (i) the investor's attempt of trading is successful only when the Poisson process $(P_t)_{t\geq 0}$ jumps,\footnote{Therefore, bigger $\lambda$ implies more frequent trading opportunities (less search frictions), on average.} and (ii) the investor needs to pay proportional transaction costs for successful trading.

 To be more specific, we consider a financial market consisting of a bond and a stock, whose price processes $(S_t^{(0)})_{t\geq0}$ and $(S_t)_{t\geq0}$ are given by the following stochastic differential equations (SDEs):
\begin{align}
dS_t^{(0)}&=S_t^{(0)} r dt\\
dS_t&=S_{t} \left( \mu dt + \sigma dB_t \right),
\end{align}
where $\mu,r,\sigma,S_0^{(0)},S_0$ are constants and $\sigma,S_0^{(0)},S_0$ are assumed to be strictly positive. 

The proportional transaction costs are described by two constants $\ue\in [0,1)$ and $\be\in [0,\infty)$: the investor gets $(1-\ue)S_t$ for selling one share of the stock, and pays $(1+\be)S_t$ for purchasing one share of the stock. 

Let $W_t^{(1)}$ be the amount of wealth invested in the stock and $W_t^{(0)}$ be the amount of wealth in the bond, at time $t\geq 0$. If the investor tries to obtain the stock worth $M_s$ at time $s\in [0,t]$, then 
\begin{equation}
\begin{split}\label{SDE_W01}
W_t^{(1)}&=w_0^{(1)}+ \int_0^t W_{s-}^{(1)} \big( \mu ds + \sigma dB_s  \big) + \int_0^t M_s dP_s,\\
W_t^{(0)}&=w_0^{(0)}+\int_0^t W_{s-}^{(0)} r ds  +\int_0^t \left( (1-\ue) M_s^-  -(1+\be)M_s^+ \right) dP_s,
\end{split}
\end{equation}
where we use notation $x^{\pm}=\max{\{0,\pm x\}}$ for $x\in \R$, and the constants $w_0^{(1)}\geq 0 $ and $w_0^{(0)}\geq 0$ are given model parameters that represent the initial position of the investor. 

The trading strategy $(M_t)_{t\geq 0}$ is called  {\it admissible} if
it is a predictable process and the corresponding total wealth process $W:=W^{(0)}+W^{(1)}$ is nonnegative all the time. This nonnegativity condition is equivalent to $W^{(0)}_t \geq 0$ and $W^{(1)}_t\geq 0$ for all $t\geq 0$, because the rebalancing times are discrete.\footnote{
Indeed, for $s>t$ and $A=\{W^{(1)}_t<0, W^{(0)}_t>0\}$ or $\{W^{(1)}_t>0, W^{(0)}_t<0\}$, we observe that
$$\PP(W_s<0 \, | \,A) \geq  \PP(P_t=P_s \textrm{  and  } W_s<0 \, | \, A)>0.$$
} 
Therefore, an admissible strategy $M$ satisfies
\begin{align}\label{M_bound}
-W_{t-}^{(1)}\leq M_t \leq \frac{W_{t-}^{(0)}}{1+\be}, \quad t\geq 0.
\end{align}

For an admissible strategy $M$ and the corresponding solutions $W^{(1)}$ and $W^{(0)}$ of the SDEs \eqref{SDE_W01}, let $X_t:=W_t^{(1)}/W_t$ be the fraction of the total wealth invested in the stock market at time $t$. Then, the inequalities in \eqref{M_bound} imply that $0\leq X_t \leq 1$ and the application of Ito's formula produces the following SDE for $(W,X)$:
\begin{equation}
\begin{split}\label{SDE_WX}
dW_t&=\left( r(1-X_{t-})+\mu X_{t-}\right) W_{t-} dt + \sigma X_{t-} W_{t-} dB_t - \left( \be M_t^+ + \ue M_t^- \right) dP_t,\\
dX_t &=X_{t-} (1-X_{t-}) \left( (\mu-r-\sigma^2 X_{t-})dt + \sigma dB_t  \right) + \left( \frac{M_t + \left(\be M_t^+ + \ue M_t^-\right) X_{t-}}{W_{t-} - \be M_t^+ -\ue M_t^-} \right) dP_t,  
\end{split}
\end{equation}
where the initial conditions are $W_0=w_0:=w_0^{(1)}+w_0^{(0)}$ and $X_0=x_0:=w_0^{(1)}/w_0$. We assume that the initial total wealth is strictly positive, $w_0>0$. The nonnegativity of $w_0^{(1)}$ and $w_0^{(1)}$ imply $0\leq x_0 \leq 1$.

Let a constant $T>0$ represent the terminal time of trading. We assume that the investor's goal is to maximize the expected log-utility of the total wealth at the terminal time.
 That is, we analyze the following optimal investment problem:
\begin{align}
\sup_{(M_t)_{t\in [0,T]}} \EE[\ln(W_T) ], \label{objective}
\end{align}
where the supremum is taken over all admissible trading strategies.

\begin{remark}
If $\lambda=\infty$ and $\be=\ue=0$, then our model becomes the classical Merton's portfolio problem. 
The case of $\be=\ue=0$ and $\lambda<\infty$ is studied in \cite{matsumoto2006}.
\end{remark}

\section{Value function}

Let $V$ be the value function of the control problem \eqref{objective}:
\begin{align}\label{value_def}
V(t,x,w)=\sup_{(M_s)_{s\in[t,T]} } \EE[ \ln(W_T)| \FF_t] \Big|_{(X_{t},W_t)=(x,w)}.
\end{align}
As usual, the scaling property of the wealth process and the property of log-utility enable us to conjecture the form of the value function as 
$$V(t,x,w)=\ln(w)+v(t,x)$$
for a function $v$ (we verify this in \thmref{verification}). Then, the HJB equation for \eqref{value_def} becomes
\begin{equation}\label{hjb_v}
\begin{split}
\begin{cases}
0=v(T,x),\\
0= v_t  + x(1-x)(\mu-r-\sigma^2 x)v_x+\frac{1}{2} \sigma^2 x^2 (1-x)^2 v_{xx}+ (\mu-r)x+r-\tfrac{1}{2}\sigma^2 x^2 - \lambda v \\
\qquad + \lambda   \sup_{y \in [0,1]}  \left( v(t,y)-  \ln \left(\frac{1+\be y}{1+\be x} \right) 1_{\{x< y\}} -\ln \left(\frac{1-\ue y}{1-\ue x} \right) 1_{\{x> y\}}    \right),
\end{cases}
\end{split}
\end{equation} 
where $v_t, v_x, v_{xx}$ are partial derivatives. The following lemma ensures that there exists a unique classical solution of the above PDE (partial differential equation).

\begin{lemma}\label{PDE_existence}
There exists a unique $v\in C([0,T]\times [0,1])\cap C^{1,2}([0,T]\times (0,1))$ that satisfies the followings:
(i) $v$ satisfies the HJB equation \eqref{hjb_v} for $(t,x)\in (0,T)\times (0,1)$.\\
(ii) For $x\in \{0,1\}$, the map $t\mapsto v(t,x)$ is continuously differentiable on $[0,T]$ and satisfies 
\begin{equation}\label{hjb_v01}
\begin{split}
\begin{cases}
0=v(T,x),\\
0= v_t(t,x)+  (\mu-r)x+r-\tfrac{1}{2}\sigma^2 x^2 - \lambda v(t,x)   \\
\qquad+   \lambda   \sup_{y\in [0,1]}  \left( v(t,y)-  \ln \left(\frac{1+\be y}{1+\be x} \right) 1_{\{x< y\}} -\ln \left(\frac{1-\ue y}{1-\ue x} \right) 1_{\{x> y\}}    \right).
\end{cases}
\end{split}
\end{equation} 
(iii) $v_t(t,x), \, x(1-x)v_x(t,x), \, x^2(1-x)^2 v_{xx}(t,x)$ are uniformly bounded on $(t,x)\in (0,T)\times (0,1)$.
\end{lemma}

\begin{proof}
See Appendix.
\end{proof}

\begin{remark}
In \lemref{PDE_existence}, we present the differential equations for $x\in (0,1)$ and $x\in \{0,1\}$ as \eqref{hjb_v} and \eqref{hjb_v01}. The reason we separate these two cases is that $v_x$ and $v_{xx}$ may not be continuously extended to the endpoints $x\in \{0,1\}$.
\end{remark}

\begin{remark}
The model with the search frictions can be seen as a generalization of \cite{dai2009finite}. However, our analysis has some technically {\it easier} features than that of \cite{dai2009finite}, because our solvency region is the first quadrant and we have no subtle issues regarding the regularity of the value function at $x=1$ or $x=0$.  
\end{remark}

\begin{theorem}\label{verification}
Let $V$ be the value function in \eqref{value_def} and $v$ be as in \lemref{PDE_existence}. Then, for $(t,x)\in [0,T]\times [0,1]$,
\begin{align}
V(t,x,w)=\ln(w)+ v(t,x).
\end{align}
\end{theorem}
\begin{proof}
Without loss of generality, we prove $V(0,x_0,w_0)=\ln(w_0)+v(0,x_0)$.
Let $M$ be an admissible trading strategy and $(W,X)$ be the corresponding solution of \eqref{SDE_WX}. Let $\tau_n:=\inf\{ t\geq 0: P_t=n\}$. If $X_{\tau_n}=0$ (resp., $X_{\tau_n}=1$), then for $\tau_n \leq t <\tau_{n+1}$, 
\begin{align}\label{X01}
X_{t}=0, \,\, dW_t = r W_t dt \quad \textrm{\big(resp.,  } X_{t}=1, \,\, dW_t = \mu W_t dt + \sigma W_t  dB_t  \big).
\end{align}

We apply Ito's formula to $\big(\ln (W_t)+v(t,X_t)\big)_{t\in [0,T]}$ and use \eqref{SDE_WX} and \eqref{X01} to obtain
\begin{displaymath}
\begin{split}
&\ln(W_T)- \ln(w_0)-v(0,x_0)\\
&= \int_0^T \bigg( \Big( v_t(t,x)  + x(1-x)(\mu-r-\sigma^2 x)v_x(t,x)+\tfrac{1}{2} \sigma^2 x^2 (1-x)^2 v_{xx}(t,x) \\
&\qquad \qquad + (\mu-r)x+r -\tfrac{1}{2}\sigma^2 x^2  \Big) \cdot  1_{\{0<x<1 \}} + \Big( v_t(t,x)  + (\mu-r)x+r -\tfrac{1}{2}\sigma^2 x^2 \Big)  \cdot 1_{\{x\in \{0,1\} \} }
 \bigg) \bigg|_{x=X_{t-}} dt\\
 &\quad + \int_0^T \sigma  \Big(x(1-x)v_x(t,x)+x\Big)\cdot 1_{\{0<x<1 \}} \bigg|_{x=X_{t-}}dB_t + \int_0^T \sigma \cdot 1_{\{x=1 \}} \bigg|_{x=X_{t-}}dB_t \\
&\quad + \sum_{ 0<t\leq T } \Big( \ln (W_t)-\ln(W_{t-}) + v(t,X_t)-v(t,X_{t-}) \Big).
\end{split}
\end{displaymath}
The stochastic integral term above is a true martingale, because \lemref{PDE_existence} (iii) implies that the integrand is uniformly bounded. The sum of jumps term above can be written as
\begin{displaymath}
\begin{split}
\int_0^T  \Big( v(t,y)- v(t,x)  -  \ln \left(\tfrac{1+\be y}{1+\be x} \right) 1_{\{x< y\}} -\ln \left(\tfrac{1-\ue y}{1-\ue x} \right)  1_{\{x> y\}}    \Big) \bigg|_{(x,y)=\left(X_{t-},\frac{X_{t-}W_{t-}+M_t}{W_{t-} - \be M_t^+ -\ue M_t^-} \right)}   dP_t.
\end{split}
\end{displaymath}
Considering that the integrand above is bounded and $(P_t - \lambda t)_{t\in [0,T]}$ is a martingale, we can write the expected value of the above expression as 
\begin{displaymath}
\begin{split}
\EE\left[\int_0^T  \lambda \Big( v(t,y)- v(t,x)  -  \ln \left(\tfrac{1+\be y}{1+\be x} \right) 1_{\{x< y\}} -\ln \left(\tfrac{1-\ue y}{1-\ue x} \right)  1_{\{x> y\}}    \Big) \bigg|_{(x,y)=\left(X_{t-},\frac{X_{t-}W_{t-}+M_t}{W_{t-} - \be M_t^+ -\ue M_t^-} \right)}   dt \right].
\end{split}
\end{displaymath}
Combining these observations, we obtain
\begin{equation}
\begin{split}\label{ito}
&\EE[\ln(W_T)]- \ln(w_0)-v(0,x_0)\\
&= \EE \Bigg[ \int_0^T \Bigg( \bigg( v_t(t,x)  + x(1-x)(\mu-r-\sigma^2 x)v_x(t,x)+\tfrac{1}{2} \sigma^2 x^2 (1-x)^2 v_{xx}(t,x)+ (\mu-r)x\\
&\qquad +r -\tfrac{1}{2}\sigma^2 x^2  + \lambda  \Big( v(t,y) -v(t,x) -  \ln \left(\tfrac{1+\be y}{1+\be x} \right) 1_{\{x< y\}} -\ln \left(\tfrac{1-\ue y}{1-\ue x} \right)  1_{\{x> y\}}    \Big) \bigg)  \cdot 1_{\{0<x<1 \}} \\
&\qquad + \bigg( v_t(t,x)  + (\mu-r)x+r -\tfrac{1}{2}\sigma^2 x^2+ \lambda  \Big( v(t,y)- v(t,x) -  \ln \left(\tfrac{1+\be y}{1+\be x} \right) 1_{\{x< y\}}\\
&\qquad  -\ln \left(\tfrac{1-\ue y}{1-\ue x} \right)  1_{\{x> y\}}    \Big) \bigg)  \cdot 1_{\{x\in \{0,1\} \} }
 \Bigg) \Bigg|_{(x,y)=\left(X_{t-},\frac{X_{t-}W_{t-}+M_t}{W_{t-} - \be M_t^+ -\ue M_t^-} \right)} dt \Bigg],
 \end{split}
\end{equation}
where the integrability is due to \lemref{PDE_existence} (iii). The equality \eqref{ito} and \lemref{PDE_existence} (i) and (ii) imply that for any admissible trading strategy $M$, 
\begin{equation}\label{sub_opt}
\EE[\ln(W_T)]\leq  \ln(w_0)+v(0,x_0).
\end{equation}

To complete the proof, we find an optimal strategy $\hat M$ that satisfies the equality in \eqref{sub_opt}. We observe that the map
\begin{align}
(t,x,y)\mapsto v(t,y)-  \ln \left(\tfrac{1+\be y}{1+\be x} \right) 1_{\{x< y\}} -\ln \left(\tfrac{1-\ue y}{1-\ue x} \right) 1_{\{x> y\}} \nonumber
\end{align}
is continuous on $[0,T]\times [0,1]^2$. Therefore, according to \lemref{meas_lem}, we can choose a measurable function $\hat y:[0,T]\times [0,1]\to [0,1]$ such that 
\begin{align}\label{hat y_def}
\hat y(t,x)\in \argmax_{y\in [0,1]}  \left( v(t,y)-  \ln \left(\tfrac{1+\be y}{1+\be x} \right) 1_{\{x< y\}} -\ln \left(\tfrac{1-\ue y}{1-\ue x} \right) 1_{\{x> y\}}\right).
\end{align}
Using $\hat y$, we define a measurable function $m : [0, T] \times [0,\infty)\times [0,1]\to \R$ as
\begin{align}\label{m_def}
m(t,w,x):=\frac{w(\hat y(t,x)-x)}{1+\be \,\hat y(t,x)} \cdot 1_{ \big\{ \hat{y}(t,x)> x \big\}} +  \frac{w(\hat y(t,x)-x)}{1-\ue  \,\hat y(t,x) }\cdot 1_{ \big\{ \hat{y}(t,x)< x \big\}}.
\end{align} 
Let $(\hat W,\hat X)$ be the unique solution\footnote{Indeed, the SDE in \eqref{SDE_WX} without $dP_t$ term has a unique (explicit) solution. The unique solution of \eqref{SDE_WX} can be obtained by patching the unique solutions on time intervals between the jump times of the Poisson process, with the jump size described by the coefficient of $dP_t$ term.}
 of SDE \eqref{SDE_WX} with $M_t=m(t, W_{t-},X_{t-})$. Now we define $\hat M_t:=m(t,\hat W_{t-},\hat X_{t-})$, then it is a predictable process.
 From \eqref{m_def} and \eqref{SDE_WX}, we observe that
$$\Delta \hat X_t = \left(\hat y(t,\hat X_{t-})-\hat X_{t-}\right) \Delta P_t.$$
Then the structure of \eqref{SDE_WX} and $0\leq \hat y \leq 1$ imply that $0\leq \hat X_t \leq 1$ and $\hat W_t \geq 0$ for all $t\in [0,T]$. Therefore, we conclude that $\hat M$ is an admissible trading strategy.

Finally, we substitute $(\hat X, \hat W, \hat M)$ for $(X,W,M)$ in \eqref{ito}. Then, the differential equations \eqref{hjb_v} and \eqref{hjb_v01} for $v$, the optimality of $\hat y$ in \eqref{hat y_def}, and the observation $\frac{\hat X_{t-}\hat W_{t-}+\hat M_t}{\hat W_{t-} - \be \hat M_t^+ -\ue \hat M_t^-}= \hat y(t,\hat X_{t-})$ produce
\begin{align}\label{opt}
\EE[\ln(\hat W_T)]= \ln(w_0)+v(0,x_0).
\end{align}
By  \eqref{sub_opt} and \eqref{opt}, we conclude $V(0,x_0, w_0)=\ln(w_0)+v(0,x_0)$ and the optimality of $\hat M$. 
\end{proof}

In the proof of \thmref{verification}, the optimal trading strategy is described by the function $\hat y$ in \eqref{hat y_def}. In the next section, we show that the maximizer in  \eqref{hat y_def} is unique, and we  characterize the form of the unique optimal strategy in detail. For this purpose, we use the {\it strict} concavity of the value function. To be more specific, we define $\tV:[0,T]\times \left([0,\infty)^2\setminus \{(0,0)\} \right)\to \R$ as
\begin{align}\label{tV_def}
\tV(t,a,b):=\sup_{(M_s)_{s\in[t,T]} } \EE \left[ \ln\left(W_T^{a,b,M}\right)   \right],
\end{align}
where $W_T^{a,b,M}$ represents the total wealth at time $T$ with $(W^{(0)}_{t},W^{(1)}_{t})=(a,b)$ and the trading strategy $M$. In other words, $\tV$ is the value function in terms of the wealth amount invested in the bond and the stock. Then, \eqref{value_def} and \eqref{tV_def} imply that
\begin{align}\label{tV_V}
\tV(t,a,b)=V\left(t,\tfrac{b}{a+b},a+b\right).
\end{align}

\begin{proposition}\label{value_concave}
For $t\in [0,T)$, the maps $(a,b)\mapsto \tV(t,a,b)$ and $x\mapsto v(t,x)$ are strictly concave.
\end{proposition}
\begin{proof}
Without loss of generality, we prove the statement for $t=0$ case. 
Let $(a_0,b_0)\neq(a_1,b_1)$ be elements of $[0,\infty)^2\setminus \{(0,0)\}$ and $\theta\in (0,1)$. According to the proof of \thmref{verification} and the relation \eqref{tV_V}, we can find optimal trading strategies for the initial positions $(a_0,b_0)$ and $(a_1,b_1)$, and we denote them by $\hat M^{0}$ and $\hat M^{1}$, respectively. 
Then the structure of the SDE in \eqref{SDE_W01} implies that $M^{\theta}:=(1-\theta) \hat M^0 +\theta \hat M^1$ is an admissible trading strategy with initial position $(a_\theta,b_\theta):=\left((1-\theta)a_0+\theta a_1,(1-\theta)b_0+\theta b_1\right)$ and satisfies\footnote{
The inequality \eqref{W_ineq} becomes strict on the event $\left\{ \omega\in \Omega: \exists t\in [0,T] \textrm{ such that }\Delta P_t(\omega)=1 \textrm{ and }\hat M_t^0(\omega) \hat M_t^1(\omega)<0\right\}$.
}
\begin{align}\label{W_ineq}
 (1-\theta) W_T^{a_0,b_0,\hat M^0}+ \theta W_T^{a_1,b_1,\hat M^1} \leq W_T^{a_\theta,b_\theta,M^\theta}.
\end{align}
We also observe that
\begin{equation}
\begin{split}\label{W_different}
\PP\left (W_T^{a_0,b_0, \hat M^0}\neq W_T^{a_1,b_1, \hat M^1}\right)&\geq\PP\left (W_T^{a_0,b_0, \hat M^0}\neq W_T^{a_1,b_1, \hat M^1} \textrm{  and  } P_T=0\right)\\
&=\PP(P_T=0)\cdot  \PP\left(W_T^{a_0,b_0, \hat M^0}\neq W_T^{a_1,b_1, \hat M^1} \, \big| \, P_T=0 \right)\\
&=e^{-\lambda T} \cdot \PP\left( (a_1-a_0) e^{rT} + (b_1-b_0)e^{(\mu-\frac{\sigma^2}{2})T + \sigma B_T} \neq 0 \right)\\
&=e^{-\lambda T}>0,
 \end{split}
\end{equation}
where the second equality is due to the independence of $B$ and $P$, and the last equality is due to $(a_0,b_0)\neq(a_1,b_1)$ and the fact that $B_T$ is a continuous random variable.
By these observations, we obtain the strict concavity of $\tV$:
\begin{displaymath}
\begin{split}
(1-\theta)\tV(0,a_0,b_0)+\theta \tV(0,a_1,b_1)&=\EE\left[ (1-\theta)\ln\left(W_T^{a_0,b_0,\hat M^0}\right) +\theta \ln\left(W_T^{a_1,b_1,\hat M^1}\right)  \right]\\
&<\EE\left[ \ln\left((1-\theta)W_T^{a_0,b_0,\hat M^0}+\theta W_T^{a_1,b_1,\hat M^1}\right)  \right]\\
&\leq \EE\left[ \ln\left(W_T^{a_\theta,b_\theta,M^\theta}\right)  \right]\\
&\leq \tV(0,a_\theta,b_\theta),
 \end{split}
\end{displaymath}
where the first inequality is due to the strict concavity of logarithm and \eqref{W_different}, and the second inequality is from \eqref{W_ineq}. Therefore, the map $(a,b)\mapsto \tV(0,a,b)$ is strictly concave. This also implies that the map $x\mapsto \tV(0,1-x,x)$ is strictly concave. Finally, \thmref{verification} and the relation \eqref{tV_V} connect $\tV$ and $v$ as
\begin{align}
\tV(0,1-x,x)=V(0,x,1)=v(0,x), \nonumber
\end{align}
and we conclude that the map $x\mapsto v(0,x)$ is strict concave.
\end{proof}

\section{Optimal strategy}

In this section, we show that the optimal strategy can be characterized in terms of the {\it no-trade region}. We start with the construction of the {\it candidate} boundary points $\uy$ and $\oy$ of the no-trade region in the following lemma.

\begin{lemma}\label{uyoy_lem}
For $t\in [0,T)$, there exist $0\leq \uy(t)\leq \oy(t)\leq 1$ such that 
\begin{displaymath}
\begin{split}
\big\{ \uy(t) \big\}=\argmax_{y\in [0,1]} \Big( v(t,y)- \ln(1+\be y) \Big),\\
\big\{ \oy(t) \big\}=\argmax_{y\in [0,1]} \Big( v(t,y)- \ln(1-\ue y) \Big).
\end{split}
\end{displaymath}
To be more specific, the following statements hold:\\
(i) The map $y\mapsto v(t,y)- \ln(1+\be y)$ strictly increases (decreses) on $y\in [0,\uy(t)]$ ($y\in [\uy(t),1]$). \\If $0<\uy(t)<1$, then $\uy(t)$ is the unique solution of the equation $ v_x(t,x)=\frac{\be}{1+\be x}$.\\
(ii) The map $y\mapsto v(t,y)- \ln(1-\ue y)$ strictly increases (decreses) on $y\in [0,\oy(t)]$ ($y\in [\oy(t),1]$). \\If $0<\oy(t)<1$, then $\oy(t)$ is the unique solution of the equation $ v_x(t,x)=-\frac{\ue}{1-\ue x}$.
\end{lemma}
\begin{proof}

Let $t\in [0,T)$ be fixed. We consider the map $z\in [0,\tfrac{1}{1+\be}]\mapsto \tV(t,1-(1+\be)z,z)$, where $\tV$ is defined in \eqref{tV_def}. \thmref{verification} and \eqref{tV_V} imply that this map is differentiable (we denote the partial derivatives as $\tV_a$ and $\tV_b$), and the map is strictly concave due to Proposition \ref{value_concave}. Therefore, its derivative 
\begin{align}\label{D_def}
D(t,z):=-(1+\be)\tV_a (t,1-(1+\be)z,z)+\tV_b (t,1-(1+\be)z,z)
\end{align}
is strictly decreasing in $z\in(0,\tfrac{1}{1+\be})$, and there exists a unique $\underline z(t) \in [0,\tfrac{1}{1+\be}]$ such that
\begin{align}\label{z_ineq}
\begin{cases}
D(t,z)>0 &\textrm{for $z\in (0,\underline z(t))$},\\
D(t,z)<0 &\textrm{for $z\in (\underline z(t),\tfrac{1}{1+\be})$}.\\
\end{cases}
\end{align}
Obviously, in case $\underline z(t)\in (0,\tfrac{1}{1+\be})$, $\underline z(t)$ is the unique solution of $D(t,z)=0$.

By \thmref{verification} and \eqref{tV_V}, we observe that for $y\in [0,1]$,
\begin{align}
v(t,y)- \ln(1+\be y)=V\left(t,y,\tfrac{1}{1+\be y} \right)=\tV \left(t,\tfrac{1-y}{1+\be y},\tfrac{y}{1+\be y}\right). \nonumber
\end{align}
We take derivative with respect to $y$ above and use \eqref{D_def} to obtain
\begin{align}\label{v and D}
\tfrac{\partial}{\partial y} \Big( v(t,y)- \ln(1+\be y) \Big)=\tfrac{1}{(1+\be y)^2} \, D\left(t,\tfrac{y}{1+\be y}  \right).
\end{align}
Now we define $\uy(t):=\frac{\underline z(t)}{1-\be \underline z(t)} \in [0,1]$.  Since the map $y\mapsto \frac{y}{1+\be y}$ is strictly increasing on $[0,1]$, the definition of $\underline z(t)$ in \eqref{z_ineq} implies that 
\begin{align}\label{uy_ineq}
\begin{cases}
 D\left(t,\frac{y}{1+\be y}  \right)>0 &\textrm{for $y\in \left(0, \uy(t) \right)$},\\
 D\left(t,\frac{y}{1+\be y}  \right)<0 &\textrm{for $y\in \left(\uy(t),1 \right)$}.\\
\end{cases}
\end{align}
Also, in case $\uy(t)\in (0,1)$, $\uy(t)$ is the unique solution of $ D\left(t,\frac{y}{1+\be y}  \right)=0$. 
From \eqref{v and D} and \eqref{uy_ineq}, we conclude that statement (i) holds and
$$\big\{ \uy(t) \big\}=\argmax_{y\in [0,1]} \Big( v(t,y)- \ln(1+\be y) \Big).$$

By the same way, we conclude that $\max_{y\in [0,1]} \big( v(t,y)- \ln(1-\ue y) \big)$ has the unique maximizer denoted by $\oy(t)$ and statement (ii) holds.

It only remains to check the inequality $\uy(t)\leq \oy(t)$. If $\oy(t)=1$, then $\uy(t)\leq \oy(t)$ is obvious. If $\oy(t)<1$, then $v_x(t,\oy(t))\leq  \frac{-\ue}{1-\ue \oy(t)}\leq \frac{\be}{1+\be \oy(t)}$ by (ii). From \eqref{v and D} and \eqref{uy_ineq}, we obtain $\uy(t)\leq \oy(t)$.
\end{proof}

In the next theorem, we explicitly characterize the optimizer $\hat y$ in \eqref{hat y_def} in terms of $\uy$ and $\oy$ in \lemref{uyoy_lem}.

\begin{theorem}\label{optimal_strategy}
For $t\in [0,T)$, the $\argmax$ in \eqref{hat y_def} is a singleton, and $\hat y$ has the following expression:
\begin{align}\label{hat y_exp}
\hat y(t,x) =  
\begin{cases}
\uy(t), & \textrm{if   }  x\in \left[0,\uy(t)\right)\\
x, & \textrm{if   }  x\in \left[\uy(t),\oy(t)\right]\\
\oy(t), & \textrm{if   }  x\in \left(\oy(t),1\right]\\
\end{cases}
\end{align}
where $\uy(t)$ and $\oy(t)$ are uniquely determined in \lemref{uyoy_lem}. 
\end{theorem}
\begin{proof}
We may rewrite the maximization in \eqref{hat y_def} as
\begin{equation}
\begin{split}\label{max_split1}
&\max_{y\in [0,1]}  \left( v(t,y)-  \ln \left(\tfrac{1+\be y}{1+\be x} \right) 1_{\{x< y\}} -\ln \left(\tfrac{1-\ue y}{1-\ue x} \right) 1_{\{x> y\}}\right)\\
&=\max \left\{ \max_{y\in [0,x]}  \left( v(t,y)-\ln \left(\tfrac{1-\ue y}{1-\ue x} \right) \right), \max_{y\in [x,1]}  \left( v(t,y)-\ln \left(\tfrac{1+\be y}{1+\be x} \right) \right) \right\}.
 \end{split}
\end{equation}
Using \lemref{uyoy_lem}, we observe that
\begin{equation}
\begin{split}\label{max_split2}
\max_{y\in [0,x]}  \left( v(t,y)-\ln \left(\tfrac{1-\ue y}{1-\ue x} \right) \right) = 
\begin{cases}
v(t,x), & \textrm{if   } x\leq \oy(t)\\
v(t,\oy(t))-\ln \left(\frac{1-\ue \oy(t)}{1-\ue x} \right),& \textrm{if   } x> \oy(t)
\end{cases}\\
\max_{y\in [x,1]}  \left( v(t,y)-\ln \left(\tfrac{1+\be y}{1+\be x} \right) \right)=
\begin{cases}
v(t,\uy(t))-\ln \left(\frac{1+\be \uy(t)}{1+\be x} \right), & \textrm{if   } x< \uy(t)\\
 v(t,x),& \textrm{if   } x \geq   \uy(t)
\end{cases}
 \end{split}
\end{equation}
where the maximizers are unique. Combining \eqref{max_split1} and \eqref{max_split2}, we obtain
\begin{equation}
\begin{split}
&\max_{y\in [0,1]}  \left( v(t,y)-  \ln \left(\tfrac{1+\be y}{1+\be x} \right) 1_{\{x< y\}} -\ln \left(\tfrac{1-\ue y}{1-\ue x} \right) 1_{\{x> y\}}\right)\\
&=\begin{cases}
v(t,\uy(t))-\ln \left(\frac{1+\be \uy(t)}{1+\be x} \right), & \textrm{if   }  x\in \left[0,\uy(t)\right)\\
v(t,x), & \textrm{if   }  x\in \left[\uy(t),\oy(t)\right]\\
v(t,\oy(t))-\ln \left(\frac{1-\ue \oy(t)}{1-\ue x} \right), & \textrm{if   }  x\in \left(\oy(t),1\right]\\
\end{cases}\nonumber
 \end{split}
\end{equation}
and conclude that the corresponding unique maximizer is as in \eqref{hat y_exp}.
\end{proof}

\begin{figure}[t]
		\begin{center}$
			\begin{array}{cc}
			\includegraphics[width=0.45\textwidth]{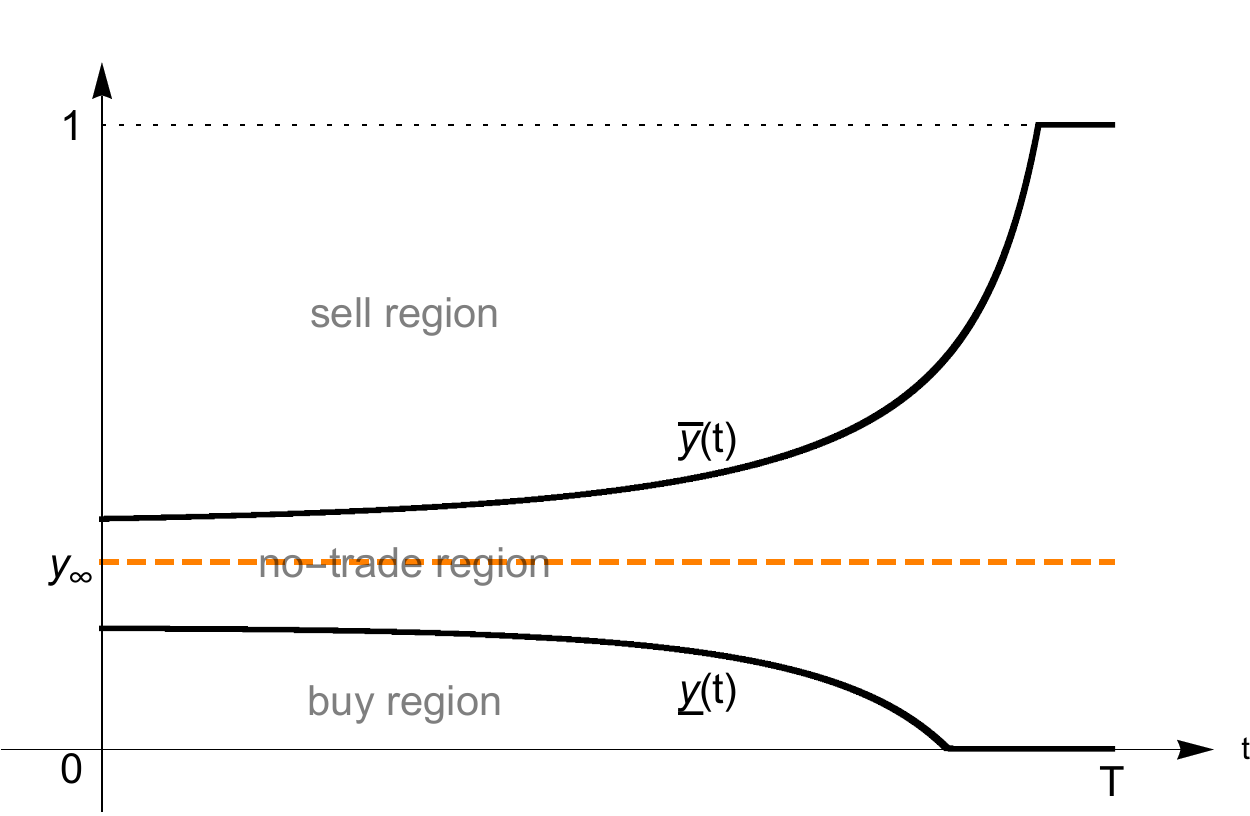} \,\,\, & \,\,\,
 			\includegraphics[width=0.45\textwidth]{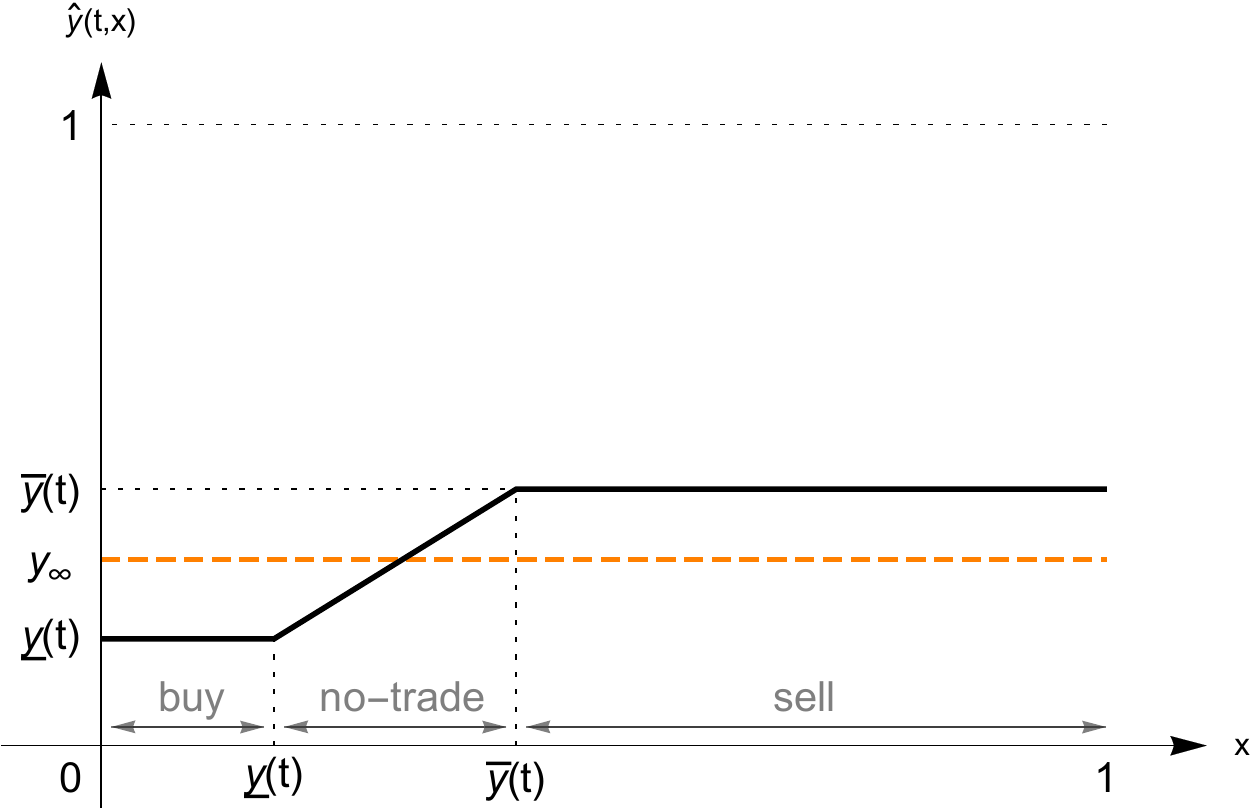} 		
			\end{array}$
	\end{center}
	 \caption{The left graph shows $\uy(t)$ and $\oy(t)$ as functions of $t$, and the right graph describes $\hat{y}(t,x)$ as a function of $x$ for fixed $t=0.5$. In both graphs, the dashed line is the Merton fraction $y_\infty=\frac{\mu - r}{\sigma^{2}}$. The parameters are $\mu=0.4, \, r=0.1, \sigma=1,\, \lambda=3,\, \ue=\be=0.05$, and $T=1$.
		}
		\label{NT}
\end{figure}

\medskip

\thmref{optimal_strategy} implies that the optimal trading strategy is characterized by the no-trade region: if the investor can trade at time $t\in [0,T)$, then the investor should minimally trade to keep the fraction of wealth invested in the stock inside of the interval $[\uy(t),\oy(t)]$. To be specific, if the fraction $X_t$ is less (more, resp.) than $\uy(t)$ ($\oy(t)$, resp.), then the investor should buy (sell, resp.) the stock and adjust the fraction to $\uy(t)$ ($\oy(t)$, resp.). If the fraction $X_t$ is inside of the interval $[\uy(t),\oy(t)]$, then the investor should not trade. Figure \ref{NT} illustrates the no-trade region and the optimal trading strategy. 

\medskip

If $\be=\ue=0$, then \lemref{uyoy_lem} implies that $\uy(t)=\oy(t)$, hence the no-trade region becomes a singleton. For $\be>0$, one may expect that the investor would not want to buy the stock at times close to the terminal time $T$ due to the transaction costs. If $\be=0$ and the Merton fraction $\frac{\mu - r}{\sigma^{2}}$ is greater than zero, then one may expect that the investor would want to hold strictly positive shares of the stock all the time. Our next task is to examine and prove this type of trading behaviors.

\medskip

For detailed analysis, we first provide stochastic representations of $v$ and $v_x$. We apply the Feynman-Kac formula (i.e., see Theorem 5.7.6 in \cite{karatzas2014brownian}) and the expression of the optimizer $\hat y$ in \thmref{optimal_strategy} to \lemref{PDE_existence}, and obtain the following representation for $v$:  
\begin{align}\label{v_FK}
v(t,x)=\int_t^T e^{-\lambda(s-t)} \EE \left[ (\mu-r)Y_s^{(t,x)}+r-\tfrac{1}{2}\sigma^2 \left(Y_s^{(t,x)} \right)^2+  \lambda \, L(s,Y_s^{(t,x)}) \right] ds,
\end{align}
where for $(s,x)\in [t,T)\times [0,1]$,
\begin{equation}
\begin{split}\label{YL_def}
Y_s^{(t,x)}&:=\frac{x \cdot \exp{\left( \left(\mu-r-\frac{1}{2}\sigma^2\right)(s-t)+\sigma(B_s-B_t)   \right)}}{x \cdot \exp{\left( \left(\mu-r-\frac{1}{2}\sigma^2\right)(s-t)+\sigma(B_s-B_t)   \right)}+(1-x) },\\
L(s,y)&:=  v(s,\hat y(s,y))-  \ln \left(\tfrac{1+\be\, \hat y(s,y)}{1+\be y} \right) 1_{\{y< \hat y(s,y)\}} -\ln \left(\tfrac{1-\ue \, \hat y(s,y)}{1-\ue y} \right) 1_{\{y> \hat y(s,y)\}}.
\end{split}
\end{equation} 

The representation of $v_x$ is given in the following lemma.

\begin{lemma}\label{vx_represent}
The function $L$ in \eqref{YL_def} is continuously differentiable with respect to $y$,
\begin{align}\label{Lx_exp}
L_y(t,x)= \begin{cases}
\frac{ \be }{1+\be x}, &x\in (0,\uy(t)]\\
 v_x(t,x), &x\in (\uy(t),\oy(t))\\
-\frac{ \ue }{1-\ue x}, & x\in [\oy(t),1)
\end{cases} \quad \textrm{for   } (t,x)\in [0,T)\times (0,1),
\end{align}
 and $v_x(t,x)$ has the following representation: for $(t,x)\in [0,T)\times (0,1)$,
\begin{align}\label{vx_exp}
v_x(t,x)=\int_t^T e^{-\lambda (s-t)} \EE \left[ \left(\tfrac{\partial}{\partial x} Y_s^{(t,x)} \right) \left( \mu-r-\sigma^2 Y_s^{(t,x)} +\lambda \, L_y(s, Y_s^{(t,x)} )  \right) \right] ds.
\end{align}
\end{lemma}
\begin{proof}
Combining \eqref{hat y_exp} and \eqref{YL_def}, we rewrite $L$ as
\begin{align}\label{L_exp}
L(t,x)= \begin{cases}
  v(t,\uy(t))-  \ln \left(\frac{1+\be\, \uy(t)}{1+\be x} \right)  , &x\in (0,\uy(t)]\\
 v(t,x), &x\in (\uy(t),\oy(t))\\
  v(t,\oy(t))-  \ln \left(\frac{1-\ue\, \oy(t)}{1-\ue x} \right) , & x\in [\oy(t),1)
\end{cases}.
\end{align}
We take derivative with respect to $x$ above and obtain the expression \eqref{Lx_exp} for $x\in (0,1)\setminus \{\uy(t),\oy(t)\}$. If $0<\uy(t)<1$ (resp., $0<\oy(t)<1$), then $v_x(t,\uy(t))=\frac{\be }{1+\be \uy(t)}$ (resp., $v_x(t,\oy(t))=-\frac{\ue }{1-\ue \oy(t)}$). Therefore, we conclude that $L$ is continuously differentiable with respect to $y$ and \eqref{Lx_exp} is valid. 

Since $-\frac{\ue }{1-\ue x}< v_x(t,x)< \frac{\be }{1+\be x}$ for $x\in  (\uy(t),\oy(t))$, we observe that for $(t,x)\in [0,T)\times (0,1)$,
\begin{align}\label{Lx_bound}
 -\frac{ \ue }{1-\ue }   \leq L_y(t,x)   \leq  \be.
\end{align}
We also observe that for $(t,x)\in [0,T)\times (0,1)$,
\begin{equation}
\begin{split}\label{YD_bound}
e^{-\left| r-\mu+\frac{\sigma^2}{2}\right| (s-t)-\sigma \left| B_s-B_t \right|} \leq \tfrac{\partial}{\partial x} Y_s^{(t,x)}   \leq  e^{\left| r-\mu+\frac{\sigma^2}{2}\right| (s-t)+\sigma \left| B_s-B_t \right|}. 
\end{split}
\end{equation}
Now we take derivative with respect to $x$ in \eqref{v_FK}. The mean value theorem and the dominated convergence theorem, together with the inequalities \eqref{Lx_bound} and \eqref{YD_bound}, allow us to take derivative inside of the expectation. By \eqref{Lx_exp} and the chain rule, we obtain the representation \eqref{vx_exp}.
\end{proof}

Using these representations, we extract some properties about the boundaries $\uy(t)$ and $\oy(t)$ of the no-trade region.

\begin{proposition}\label{uyoy_prop}
{\bf (Properties of no-trade region)} \\
Let $y_{\infty} := \frac{\mu - r}{\sigma^{2}}$ denote the Merton fraction.  \\
(i) If $\be>0$ ($\ue>0$, resp.), then there exists $t_0\in [0,T)$ such that $\uy(t)=0$ ($\oy(t)=1$, resp.) for $t\in [t_0,T)$. \\
(ii) If $0<y_\infty$ ($y_\infty<1$, resp.), then $\oy(t)>0$ ($\uy(t)<1$, resp.) for $t\in [0,T)$.\\
(iii) If $0<y_\infty<1$ and at least one of $\be$ and $\ue$ is strictly positive, then $\uy(t)<\oy(t)$ for $t\in [0,T)$.\\
(iv) If $0<y_\infty$ and $\be=0$ ($y_\infty<1$ and $\ue=0$, resp.), then $\uy(t)>0$ ($\oy(t)<1$, resp.) for $t\in [0,T)$. 
\end{proposition}
\begin{proof}
By \eqref{YD_bound}, we have
\begin{align}\label{EY_bound}
0<\EE \left[ \tfrac{\partial}{\partial x} Y_s^{(t,x)} \right]\leq 2 e^{(|r-\mu|+\sigma^2)(s-t)},
\end{align}
and we apply this inequality and \eqref{Lx_bound} to the expression \eqref{vx_exp} to obtain 
\begin{align}\label{vx_bounds}
\underline c \, (T-t) \leq v_x(t,x)\leq  \bar c \, (T-t)  \quad \textrm{for   }x\in (0,1),
\end{align}
where $\underline c$ and $\bar c$ are constants that only depend on $\mu,r,\sigma, \lambda,\be, \ue, T$. 

(i) Suppose that $\be>0$ (the case of $\ue>0$ can be treated similarly). If $\bar c \leq 0$, then we choose $t_0=0$ and easily observe that
\begin{align}\label{uy=0}
\lim_{x\downarrow 0} \tfrac{\partial}{\partial x} \Big( v(t,x)- \ln(1+\be x) \Big) \leq 0 \quad \textrm{for  }t\in [t_0,T).
\end{align}
If $\bar c>0$, then \eqref{vx_bounds} implies that \eqref{uy=0} holds with $t_0=\left(T- \frac{\be}{\bar c }\right)^+$. Therefore, in any case, we can choose $t_0\in [0,T)$ that satisfies \eqref{uy=0}. By \lemref{uyoy_lem} and \eqref{uy=0}, we conclude $\uy(t)=0$ for $t\in [t_0,T)$. 


\medskip

(ii) Suppose that  $0<y_\infty$ (the case of $y_\infty<1$ can be treated similarly). Let $t\in [0,T)$ be fixed. To check $\oy(t)>0$, we observe from \eqref{vx_exp} and the dominated convergence theorem that
\begin{equation}
\begin{split}\label{vx_lim_0}
\lim_{x\downarrow 0}  v_x(t,x) &= \int_t^T e^{-\lambda (s-t)} \EE \left[\lim_{x\downarrow 0} \left(\tfrac{\partial}{\partial x} Y_s^{(t,x)} \right) \left( \mu-r-\sigma^2 Y_s^{(t,x)} + \lambda L_y(s, Y_s^{(t,x)} )  \right) \right] ds\\
&=\int_t^T e^{-\lambda (s-t)} \EE \left[    e^{(\mu-r-\tfrac{1}{2}\sigma^2)(s-t)+\sigma (B_s-B_t)} \right] \\
&\qquad \cdot \left( \mu-r+ \lambda \Big(\be \cdot 1_{\left\{ \uy(s)>0 \right\}}+\lim_{x\downarrow 0} v_x(s,x) \cdot 1_{\left\{ \uy(s)=0<\oy(s) \right\}}-\ue \cdot 1_{\left\{ \oy(s)=0 \right\}}\Big) \right)   ds\\
&\geq \int_t^T e^{(\mu-r-\lambda) (s-t)} \left( \mu-r - \lambda \ue \right) ds,
\end{split}
\end{equation}
where the second equality is from \eqref{Lx_exp} and \eqref{YD_bound}, and the inequality is due to the fact that $-\frac{\ue }{1-\ue x}< v_x(t,x)$ for $x<\oy(t)$. 
Obviously, the last integral in  \eqref{vx_lim_0} is nonnegative if $\mu-r-\lambda \ue\geq 0$. If  $\mu-r-\lambda \ue< 0$, then $\mu-r-\lambda < 0$ due to $\ue\in [0,1)$ and we observe that
$$
 \int_t^T e^{(\mu-r-\lambda) (s-t)} \left( \mu-r - \lambda \ue \right) ds =
 \frac{\mu-r-\lambda \ue}{\mu-r-\lambda}\left( e^{(\mu-r-\lambda)(T-t)}-1\right) > - \ue,
$$
where we use $y_\infty>0$ for the inequality. The above inequality and \eqref{vx_lim_0} imply 
$$
\lim_{x\downarrow 0} \tfrac{\partial}{\partial x} \Big( v(t,x)- \ln(1-\ue x) \Big) > 0 ,
$$
and we conclude $\oy(t)>0$ by \lemref{uyoy_lem}.


\medskip

(iii) Suppose that  $0<y_\infty<1$ and $\be>0$ (the case of $\ue>0$ can be treated similarly). According to (ii), $\uy(t)<1$ and $\oy(t)>0$. Therefore, there are only two possibilities: $\uy(t)=0$ or $0<\uy(t)<1$. In case $\uy(t)=0$, we immediately obtain $\uy(t)<\oy(t)$ since $\oy(t)>0$. In case $0<\uy(t)<1$, by \lemref{uyoy_lem}, we have 
$$
v_x(t,\uy(t))=\frac{\be}{1+\be\, \uy(t)}>\frac{-\ue}{1-\ue\, \uy(t)},
$$
and this inequality, together with \eqref{v and D} and \eqref{uy_ineq}, implies $\uy(t)< \oy(t)$.

\medskip

(iv) Suppose that $0<y_\infty$ and $\be=0$ (the case of $y_\infty<1$ and $\ue=0$ can be treated similarly). \lemref{uyoy_lem} implies that 
\begin{align}\label{uy0_iff}
\uy(t)>0 \quad \textrm{if and only if} \quad \lim_{x\downarrow 0} v_x(t,x)>0.
\end{align}
The second equality in \eqref{vx_lim_0} can be written as
\begin{align}\label{vx_lim_0_simple}
\lim_{x\downarrow 0}  v_x(t,x) &= \int_t^T e^{(\mu-r-\lambda) (s-t)} \left( \mu-r+ \lambda \lim_{x\downarrow 0} v_x(s,x) \cdot 1_{\left\{ \uy(s)=0<\oy(s) \right\}} \right)   ds,
\end{align}
where we substitute $\be=0$ and use $\oy(s)>0$ by (ii). Define $t^*$ as
\begin{align}\label{t*_def}
t^*:=\inf\left\{ t\in [0,T): \,\, \lim_{x\downarrow 0 } v_x(s,x) > 0 \textrm{  for all  }s\in [t,T) \right \}.
\end{align}
Due to \eqref{vx_bounds} and the strict positivity of $\mu-r$, for $t$ close enough to $T$, the integrand in \eqref{vx_lim_0_simple} is strictly positive and $\lim_{x\downarrow 0}  v_x(t,x)>0$.  Therefore, the set in \eqref{t*_def} is non-empty and $t^*<T$. Suppose that $t^*>0$. For $0\leq \delta\leq t^*$, we apply \eqref{t*_def} and \eqref{uy0_iff} to expression \eqref{vx_lim_0_simple} and obtain
\begin{align}
\lim_{x\downarrow 0}  v_x(t^*-\delta,x) &> \int_{t^*}^T e^{(\mu-r-\lambda) (s-t^*+\delta)} \left( \mu-r \right)   ds \nonumber\\
&\quad + \int_{t^*-\delta}^{t^*} e^{(\mu-r-\lambda) (s-t^*+\delta)} \left( \mu-r+ \lambda \lim_{x\downarrow 0} v_x(s,x) \cdot 1_{\left\{ \uy(s)=0<\oy(s) \right\}} \right)ds. \nonumber
\end{align}
The first integral above is greater than a strictly positive number independent of $\delta$, and the second integral can be made arbitrary close to zero as $\delta \downarrow 0$, due to \eqref{vx_bounds}. Therefore, there exists $\delta^*\in (0,t^*]$ such that $\lim_{x\downarrow 0}  v_x(t^*-\delta,x) >0$ for all $\delta \in [0,\delta^*]$. Then, $\lim_{x\downarrow 0}  v_x(s,x) >0$ for all $s\in [t^*-\delta^*,T)$, and this contradicts to the definition of $t^*$ in \eqref{t*_def}. Therefore, we conclude that $t^*=0$, and now \eqref{uy0_iff} implies that $\uy(t)>0$ for $t\in (0,T)$. Lastly, $\uy(0)>0$ is obtained by \eqref{uy0_iff} and \eqref{vx_lim_0_simple}.
\end{proof}

\begin{remark}
Straightforward interpretations of \proref{uyoy_prop} are as follows: \\
(i) The existence of the transaction costs for selling (buying, resp.) the stock makes the investor not to sell (buy, resp.) the stock when it is close to the terminal time. For short period of time, the benefit of rebalancing is small.\\
(ii) If $0<y_\infty$ ($y_\infty<1$, resp), then the investor never rebalances to the zero-holding of the stock (bond, resp.). However, if the initial holding of the stock (bond, resp.) is zero, then the investor may not try to leave the state of zero-holding of the stock (bond, resp.), depending on the size of the transaction costs.\\
(iii) If $0<y_\infty<1$ and there exist transaction costs, then the no-trade interval is non-trivial (with strictly positive length) all the time. \\
(iv) If $0<y_\infty$ ($y_\infty<1$, resp.) and there is no cost for buying (selling, resp.) the stock, then just holding the bond (stock, resp.) and setting zero balance in the stock (bond, resp.) is suboptimal, even with the search frictions and transaction costs.
\end{remark}

\begin{remark}\label{merton_out}
Obviously, if $y_\infty\leq 0$ ($y_\infty\geq 1$, resp.), then there is no reason for buying (selling, resp.) the stock, so $\uy(t)=0$ ($\oy(t)=1$, resp.) for $t\in [0,T)$.
\end{remark}

\section{Asymptotic analysis}

In this section, we provide asymptotic analysis for small transaction costs. To be specific, we focus on the first order approximation of the no-trade region and the value function with respect to the transaction cost parameter around zero. To consider non-trivial cases (see \remref{merton_out}), we assume that the Merton fraction $y_\infty$ is between zero and one, and we set $\be=\ue$ for convenience.

\begin{assumption}\label{assumption}
In this section, we assume that $0 < y_{\infty} < 1$ and $\be = \ue = \epsilon$ for $\epsilon \in [0,1)$.
\end{assumption}

\begin{notation}{\bf (Current section only)}\\
(i) To emphasize their dependence on the transaction cost parameter $\epsilon$, we denote $v,\uy,\oy,\hat y,L,L_y$ by $v^\epsilon,\uy^\epsilon,\oy^\epsilon,\hat y^\epsilon, L^\epsilon, L_y^\epsilon$. In particular, when $\epsilon=0$, they are denoted by $v^0,\uy^0,\oy^0,\hat y^0,L^0,L_y^0$. \\
(ii) Under \assref{assumption}, \lemref{uyoy_lem} and \proref{uyoy_prop} imply that
\begin{align}\label{y0}
0<\hat y^0(t,x)=\uy^0(t)=\oy^0(t)<1 \quad \textrm{for} \quad (t,x)\in [0,T)\times [0,1].
\end{align}
In words, $\hat y^0(t,x)$ is independent of variable $x$ (just a function of $t$) and its value equals $\uy^0(t)$ and $\oy^0(t)$. For convenience, we abuse notation and write $\hat y^0(t)$ for $\hat y^0(t,x)$ (i.e., $\hat y^0(t)=\uy^0(t)=\oy^0(t)$).
\end{notation}

We start with the technical lemma that is used in the proof of the asymptotic result.

\begin{lemma}\label{tech_asymptotitcs}
(i) Let $\epsilon=0$. For $t\in [0,T)$, 
\begin{align}
&\quad v_x^0(t,\hat y^0(t))=0, \label{vx0_bound}\\
&\sup_{x\in (0,1)}  v_{xx}^0(t,x) <0. \label{gxx>0_all}
\end{align}
(ii) Let $F:[0,T)\times (0,1) \to \R$ be defined as
\begin{align}\label{F_def}
F(t,x):=\lambda \int_t^T e^{-\lambda(s-t)} \EE \left[ \left(\tfrac{\partial}{\partial x} Y_s^{(t,x)} \right) \cdot \sgn\left( \hat y^0(s)-Y_s^{(t,x)} \right)  \right] ds.
\end{align}
Then, for $t\in [0,T)$, 
\begin{align}\label{F_ineq}
-1<F(t,\hat y^0(t))<1 .
\end{align}
(iii) Suppose that $x_0\in (0,1)$ and  $\,\,\lim_{\epsilon \downarrow 0} x_\epsilon=x_0$. Then, for $t\in [0,T)$, 
\begin{align}
\lim_{\epsilon \downarrow 0} v_{x}^\epsilon (t,x_\epsilon) &= v_{x}^0 (t, x_0), \label{vx_uni_conv}\\
\lim_{\epsilon \downarrow 0} v_{xx}^\epsilon (t,x_\epsilon) &= v_{xx}^0 (t, x_0). \label{vxx_uni_conv}
\end{align}
\end{lemma}

\begin{proof}
See Appendix.
\end{proof}

The following theorem provides the first order approximation of the no-trade boundaries.

\begin{theorem}\label{asymptotic_y}
For $t\in [0,T)$, 
\begin{align}
\uy^{\epsilon} (t) = \hat y^0(t) - \frac{F(t,\hat y^0(t))-1 }{v_{xx}^0(t,\hat y^0(t))} \cdot \epsilon + o(\epsilon),\label{uy_asymp}\\
\oy^{\epsilon} (t) =\hat y^0(t)  - \frac{F(t,\hat y^0(t))+1 }{v_{xx}^0(t,\hat y^0(t))} \cdot \epsilon + o(\epsilon),\label{oy_asymp}
\end{align}
where $F$ is defined in \eqref{F_def}. In particular, for small enough $\epsilon>0$, we have 
\begin{align}\label{y_orders}
0<\uy^{\epsilon} (t) <\hat y^0(t)<\oy^{\epsilon} (t)<1.
\end{align}
\end{theorem}
\begin{proof}

For $(t,x)\in [0,T)\times (0,1)$, the expression \eqref{vx_exp} implies that
\begin{equation}
\begin{split}\label{vx_diff_bound}
\left |v_x^\epsilon(t,x)-v_x^0(t,x) \right |  &= \left |   \lambda  \int_t^T e^{-\lambda (s-t)} \EE \left[ \left(\tfrac{\partial}{\partial x} Y_s^{(t,x)} \right)  L_y^\epsilon(s, Y_s^{(t,x)} )  \right] ds   \right | \\
&\leq  \lambda  \int_t^T e^{-\lambda (s-t)} \EE \left[ \left(\tfrac{\partial}{\partial x} Y_s^{(t,x)} \right)  \right] ds  \cdot \tfrac{\epsilon}{1-\epsilon}\\
&\leq \lambda  T \cdot e^{(\mu-r+\sigma^2)T}  \cdot \tfrac{\epsilon}{1-\epsilon},
\end{split}
\end{equation}
where the first inequality is due to \eqref{Lx_bound} and the positivity of $\frac{\partial}{\partial x} Y_s^{(t,x)}$ in \eqref{YD_bound}, and the second inequality is due to \eqref{EY_bound} and \assref{assumption}.

\proref{uyoy_prop} and \assref{assumption} imply that $\uy^\epsilon(t)<1$ and $\oy^\epsilon(t)>0$ for $t\in [0,T)$. By \lemref{uyoy_lem}, if $0<\uy^\epsilon(t)<1$, then $ v_x^\epsilon(t, \uy^\epsilon(t))= \frac{\epsilon}{1+\epsilon \cdot \uy^\epsilon(t)}$, and if $\uy^\epsilon(t)=0$, then $-\frac{\epsilon}{1-\epsilon x}<v_x^\epsilon(t, x)<\frac{\epsilon}{1+\epsilon x}$ for $x\in (0,\oy^\epsilon(t))$. In any case, we have
 \begin{align}\label{vxe_bound}
  \left| v_x^\epsilon(t, \uy^\epsilon(t)) \right | \leq \tfrac{\epsilon}{1-\epsilon}, 
 \end{align}
 where $v_x^\epsilon(t, 0):=\lim_{x\downarrow 0} v_x^\epsilon(t,x)$ is well-defined (see \eqref{vx_lim_0} for details) for the case of $\uy^\epsilon(t)=0$. 
We use the mean value theorem and \eqref{vx0_bound} to obtain
\begin{equation}
\begin{split}
 \inf_{x\in (0,1)} \left | v_{xx}^0(t,x) \right | \cdot  \left|\uy^\epsilon(t) - \hat y^0(t) \right| & \leq \left |v_x^0(t,\uy^\epsilon(t))-v_x^0(t,\hat y^0(t)) \right |   \\
&\leq  \left |v_x^0(t,\uy^\epsilon(t))-v_x^\epsilon(t, \uy^\epsilon(t)) \right | + \left| v_x^\epsilon(t, \uy^\epsilon(t)) \right | .
\end{split}
\end{equation}
The above inequality, together with \eqref{gxx>0_all}, \eqref{vx_diff_bound}, and \eqref{vxe_bound}, we conclude that 
\begin{align}\label{uy_lim}
\uy^\epsilon(t) - \hat y^0(t) = O(\epsilon) \quad \textrm{for} \quad t\in [0,T).
\end{align}
By the same way, we also obtain
\begin{align}\label{oy_lim}
\oy^\epsilon(t) - \hat y^0(t) = O(\epsilon)  \quad \textrm{for} \quad t\in [0,T).
\end{align}

The expression of $L_y^\epsilon$ in \eqref{Lx_exp}, together with \eqref{uy_lim} and \eqref{oy_lim}, implies the following limit:
\begin{align}\label{Lx/e_lim}
\lim_{\epsilon \downarrow 0}\frac{L_y^\epsilon(t,x)}{\epsilon} = \begin{cases}
1, &\textrm{if  }x\in (0,\hat y^0(t))\\
-1, &\textrm{if  }x\in (\hat y^0(t),1)\\
\end{cases} \quad \textrm{for}\quad t\in [0,T).
\end{align}
For $(s,x)\in (t,T)\times (0,1)$, the observation $\PP\left(Y_s^{(t,x)}=\hat y^0(t)\right)=0$ and \eqref{Lx/e_lim} produce
\begin{align}\label{Lx/eY_lim}
\lim_{\epsilon\downarrow 0 } \frac{L_y^\epsilon(s,Y_s^{(t,x)})}{\epsilon} =  \sgn\left( \hat y^0(t)- Y_s^{(t,x)} \right) \quad \textrm{almost surely.}
\end{align}
For $(t,x)\in [0,T)\times (0,1)$, the limit \eqref{Lx/eY_lim} and the inequalities \eqref{Lx_bound} and \eqref{YD_bound} enable us to use the dominated convergence theorem to obtain
\begin{equation}
\begin{split}\label{vxe_conv}
\frac{v_x^\epsilon (t, x)-v_x^0 (t, x)}{\epsilon} &=   \lambda  \int_t^T e^{-\lambda (s-t)} \EE \left[ \left(\tfrac{\partial}{\partial x} Y_s^{(t,x)} \right)  \tfrac{L_y^\epsilon(s, Y_s^{(t,x)} )}{\epsilon}  \right] ds    \xrightarrow[\epsilon \downarrow 0]{} F(t,x).
\end{split}
\end{equation}

The inequality $0<\hat y^0(t)<1$ and \eqref{uy_lim} imply that $0<\uy^\epsilon(t)<1$ for small enough $\epsilon>0$. Hence, by \lemref{uyoy_lem}, we obtain
\begin{align}\label{vx_small_e}
v_x^\epsilon(t, \uy^\epsilon(t))= \frac{\epsilon}{1+\epsilon \, \uy^\epsilon(t)} \quad \textrm{for small enough $\epsilon>0$.} \quad 
\end{align}
By the mean value theorem, there exists $k(\epsilon)$ such that
\begin{align}\label{mvt_vx}
v_x^\epsilon (t,\uy^\epsilon(t))-v_x^\epsilon(t,\hat y^0(t))=v_{xx}^\epsilon (t,k(\epsilon))\left(\uy^{\epsilon} (t) -\hat y^0(t)\right)  \quad \textrm{and} \quad \lim_{\epsilon \downarrow 0} k(\epsilon)=\hat y^0(t).
\end{align}
Now we obtain \eqref{uy_asymp} as follows:
\begin{equation}
\begin{split}
\lim_{\epsilon \downarrow 0}\frac{\uy^{\epsilon} (t) -\hat y^0(t)}{\epsilon} &= \lim_{\epsilon \downarrow 0}\frac{v_x^\epsilon (t,\uy^\epsilon(t))-v_x^\epsilon(t,\hat y^0(t))}{\epsilon \cdot v_{xx}^\epsilon (t,k(\epsilon))}\\
&=\lim_{\epsilon \downarrow 0} \frac{1}{ v_{xx}^\epsilon (t,k(\epsilon))} \left( \frac{1}{1+\epsilon \, \uy^\epsilon (t)} - \frac{v_x^\epsilon (t, \hat y^0(t))-v_x^0 (t, \hat y^0(t))}{\epsilon}    \right)\\
&= - \frac{F(t,\hat y^0(t))-1 }{ v_{xx}^0 (t,\hat y^0(t))} ,
\end{split}
\end{equation}
where the first equality is from \eqref{mvt_vx}, the second equality is due to \eqref{vx_small_e} and \eqref{vx0_bound}, and the third equality is due to  \eqref{vxe_conv} and \eqref{vxx_uni_conv} with the limit in \eqref{mvt_vx}. 
We also obtain \eqref{oy_asymp} by the same way.

Due to \eqref{gxx>0_all} and \eqref{F_ineq}, we observe that $- \frac{F(t,\hat y^0(t))-1 }{v_{xx}^0(t,\hat y^0(t))}<0$ and $ - \frac{F(t,\hat y^0(t))+1 }{v_{xx}^0(t,\hat y^0(t))}>0$. Therefore, \eqref{uy_asymp} and \eqref{oy_asymp} imply that \eqref{y_orders} holds for small enough $\epsilon>0$.
\end{proof}



\begin{figure}[t]
		\begin{center}$
			\begin{array}{cc}
			\includegraphics[width=0.45\textwidth]{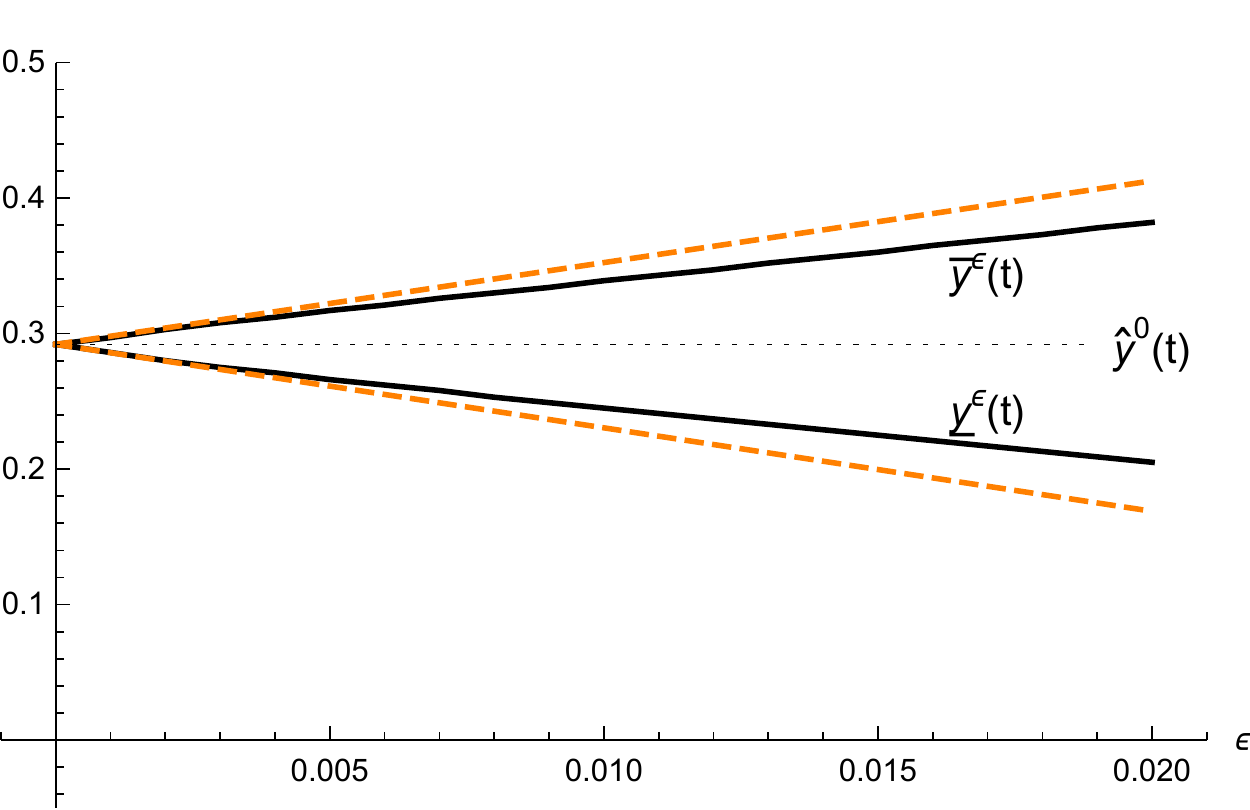} \,\,\, & \,\,\,
 			\includegraphics[width=0.45\textwidth]{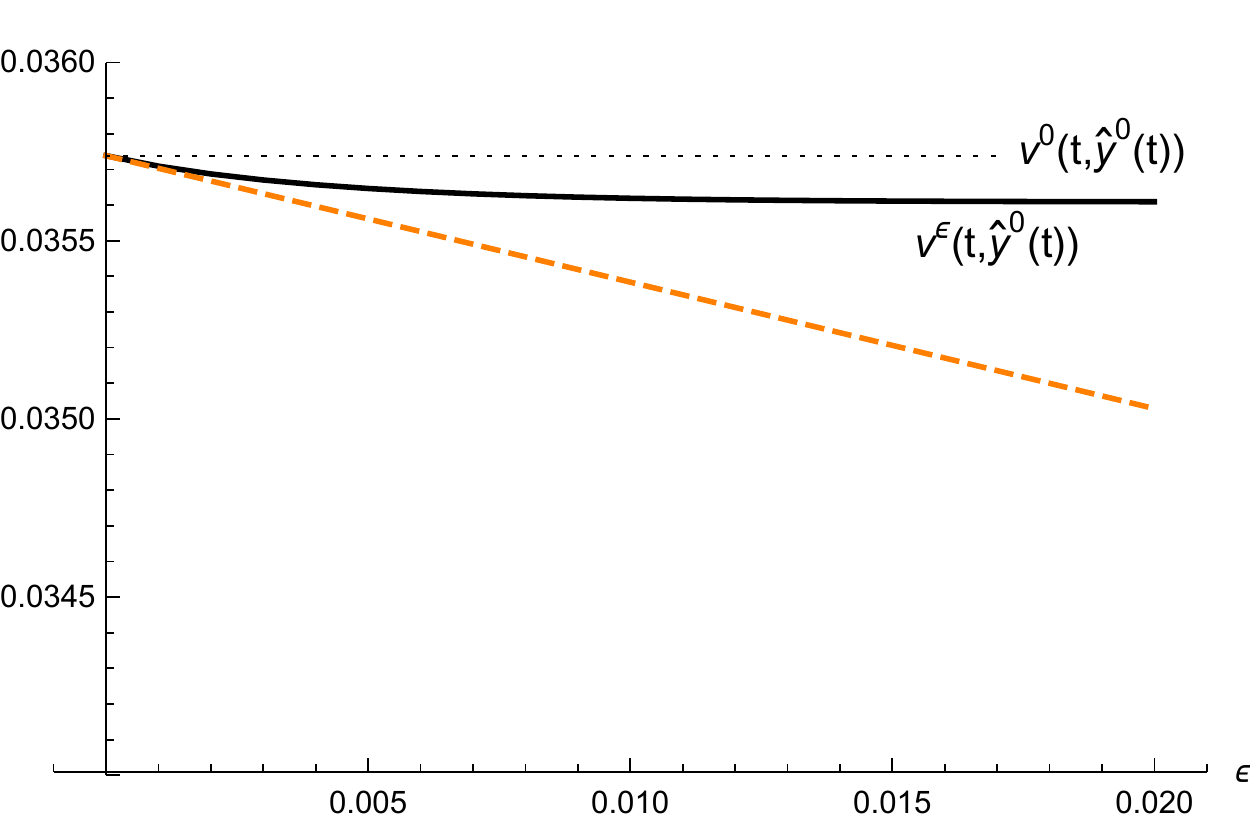} 		
			\end{array}$
	\end{center}
	 \caption{The left graph shows $\uy^\epsilon(t)$ and $\oy^\epsilon(t)$ as functions of $\epsilon$, where the dashed lines are the linear approximations of them in \thmref{asymptotic_y}, i.e., $\hat y^0(t) - \frac{F(t,\hat y^0(t))-1 }{v_{xx}^0(t,\hat y^0(t))} \cdot \epsilon$ and $\hat y^0(t) - \frac{F(t,\hat y^0(t))-1 }{v_{xx}^0(t,\hat y^0(t))} \cdot \epsilon$.
	 The right graph describes $v^\epsilon(t,\hat{y}^0(t))$ as a function of $\epsilon$, where the dashed line is the linear approximation of it in \thmref{asymptotic_v}, i.e., $v^0(t,\hat y^0(t)) - \left( G(t) + \lambda \int_t^T G(s) ds \right) \cdot \epsilon$.	 
	 In both graphs, the parameters are $\mu=0.4, \, r=0.1, \sigma=1,\, \lambda=3, \, T=1$, and $t=0.75$.
		}
		\label{asymp_fig}
\end{figure}

The following theorem provides the first order approximation of the value function.

\begin{theorem}\label{asymptotic_v}
For $t\in [0,T)$, 
\begin{align}
v^\epsilon(t,\hat y^0(t))=v^0(t,\hat y^0(t)) - \left( G(t) + \lambda \int_t^T G(s) ds \right) \cdot \epsilon + o(\epsilon),\label{v_asymp}
\end{align}
where $G:[0,T]\to \R$ is defined as
\begin{align}\label{G_def}
G(t):=\lambda \int_t^T e^{-\lambda(s-t)} \EE \left[ \left| Y_s^{(t,\hat y^0(t))}- \hat y^0(s) \right| \right] ds
\end{align}
\end{theorem}
\begin{proof}
For $\epsilon>0$, we define a function $\psi^\epsilon:[0,T) \rightarrow \R$ as
\begin{align}\label{psi_def}
\psi^\epsilon(t):=\frac{v^\epsilon(t,\hat y^0(t))-v^0(t,\hat y^0(t))}{\epsilon}.
\end{align} 
Using \eqref{v_FK} and \eqref{L_exp}, we obtain
\begin{equation}
\begin{split}\label{psi_ode}
\psi^\epsilon(t) &= \lambda \int_t^T e^{-\lambda(s-t)}  \psi^\epsilon(s) ds + \lambda \int_t^T e^{-\lambda(s-t)} \EE\left[ I_1^{(t,s,\epsilon)}+I_2^{(t,s,\epsilon)}+I_3^{(t,s,\epsilon)}
   \right] ds,
\end{split}
\end{equation}
where the random variables $I_1^{(t,s,\epsilon)}, I_2^{(t,s,\epsilon)},I_3^{(t,s,\epsilon)}$ are defined as
\begin{equation}
\begin{split}\label{I_def}
I_1^{(t,s,\epsilon)}&:=\left( \tfrac{v^\epsilon(s,\uy^\epsilon(s))-v^\epsilon(s,\hat y^0(s))}{\epsilon} - \frac{ 1}{\epsilon}\ln\left( \tfrac{1+\epsilon \, \uy^\epsilon(s)}{1+\epsilon \, Y_s^{(t,\hat y^0(t))}} \right) \right) \cdot 1_{\left\{ Y_s^{(t,\hat y^0(t))}\leq \uy^\epsilon (s)  \right \} },\\
I_2^{(t,s,\epsilon)}&:=\left( \tfrac{v^\epsilon(s,Y_s^{(t,\hat y^0(t))})-v^\epsilon(s,\hat y^0(s))}{\epsilon}\right)\cdot 1_{\left\{ \uy^\epsilon (s)<Y_s^{(t,\hat y^0(t))} <\oy^\epsilon(s)  \right \} },\\
I_3^{(t,s,\epsilon)}&:=\left( \tfrac{v^\epsilon(s,\oy^\epsilon(s))-v^\epsilon(s,\hat y^0(s))}{\epsilon} - \frac{ 1}{\epsilon}\ln\left( \tfrac{1-\epsilon \, \oy^\epsilon(s)}{1-\epsilon \, Y_s^{(t,\hat y^0(t))}} \right) \right) \cdot 1_{\left\{  \oy^\epsilon (s) \leq Y_s^{(t,\hat y^0(t))}   \right \} }.
\end{split}
\end{equation}
The mean value theorem produces 
\begin{align}\label{ln_bound}
\left| \tfrac{ 1}{\epsilon}\ln\left( \tfrac{1+\epsilon \, \uy^\epsilon(s)}{1+\epsilon \, x} \right)\right| \leq 1 \quad \textrm{and} \quad \left|  \tfrac{ 1}{\epsilon}\ln\left( \tfrac{1-\epsilon \, \oy^\epsilon(s)}{1-\epsilon \,x} \right)\right| \leq \tfrac{1}{1-\epsilon} \quad \textrm{for  } x\in (0,1).
\end{align}
By \assref{assumption} and \proref{uyoy_prop} (ii), we obtain $\uy^\epsilon(s)<1$ and $\oy^\epsilon(s)>0$ for $s\in [0,T)$. The definition of $Y_s^{(t,x)}$ in \eqref{YL_def} indicates that $0<Y_s^{(t,x)}<1$. Therefore,
\begin{align}\label{indicator_ineq}
1_{\left\{ Y_s^{(t,\hat y^0(t))}\leq \uy^\epsilon (s)  \right \} } \leq 1_{\big\{ 0<\uy^\epsilon (s)<1  \big \} }, \quad 1_{\left\{  \oy^\epsilon (s) \leq Y_s^{(t,\hat y^0(t))}   \right \} } \leq 1_{\big\{ 0<\oy^\epsilon (s)<1  \big \} }.
\end{align}
The concavity of $v^\epsilon$ on $x$ variable (see \proref{value_concave}) and \eqref{ln_bound} and \eqref{indicator_ineq} imply
\begin{equation}
\begin{split}\label{I123_ineq1}
\left| I_1^{(t,s,\epsilon)} \right| 
&\leq \left(\max\left\{ \left| \tfrac{v_x^\epsilon(s,\uy^\epsilon(s))}{\epsilon}\right| , \, \left| \tfrac{v_x^\epsilon(s,\hat y^0(s))}{\epsilon}\right|  \right\}\cdot \left| \uy^\epsilon(s)-\hat y^0(s) \right|  + 1 \right) \cdot 1_{\left\{ 0<\uy^\epsilon (s)<1  \right \} },\\
\left| I_2^{(t,s,\epsilon)} \right| 
&\leq \max\left\{ \left| \tfrac{v_x^\epsilon(s,Y_s^{(t,\hat y^0(t))})}{\epsilon}\right| , \, \left| \tfrac{v_x^\epsilon(s,\hat y^0(s))}{\epsilon}\right|  \right\}    \cdot \left| Y_s^{(t,\hat y^0(t))}-\hat y^0(s) \right|  \cdot 1_{\left\{ \uy^\epsilon (s)<Y_s^{(t,\hat y^0(t))} <\oy^\epsilon(s)  \right \} },\\
\left| I_3^{(t,s,\epsilon)} \right| 
&\leq \left(\max\left\{ \left| \tfrac{v_x^\epsilon(s,\oy^\epsilon(s))}{\epsilon}\right| , \, \left| \tfrac{v_x^\epsilon(s,\hat y^0(s))}{\epsilon}\right|  \right\}\cdot \left| \oy^\epsilon(s)-\hat y^0(s) \right| + \frac{1}{1-\epsilon}\right) \cdot 1_{\left\{ 0<\oy^\epsilon (s)<1  \right \} }.
\end{split}
\end{equation}
Combining \eqref{vx_diff_bound} and \eqref{vx0_bound}, we obtain
\begin{align}\label{vx_y0_bound}
\left| v_x^\epsilon(s,\hat y^0(s))\right| \leq \lambda  T \cdot e^{(\mu-r+\sigma^2)T}  \cdot  \tfrac{\epsilon}{1-\epsilon}.
\end{align}
By \lemref{uyoy_lem}, we have the following inequalities:
\begin{equation}
\begin{split}\label{vx_ye_bound}
&\left| v_x^\epsilon(s,\uy^\epsilon(s))\right| \cdot 1_{\left\{ 0<\uy^\epsilon (s)<1  \right \} } \leq \epsilon ,
\quad \left| v_x^\epsilon(s,\oy^\epsilon(s))\right| \cdot  1_{\left\{ 0<\oy^\epsilon (s)<1  \right \} }\leq   \tfrac{\epsilon}{1-\epsilon},\\
&\left| v_x^\epsilon(s,x)\right|   \cdot 1_{\left\{ \uy^\epsilon (s)<x<\oy^\epsilon(s)  \right \} } \leq \tfrac{\epsilon}{1-\epsilon}\quad \textrm{for} \quad x\in (0,1).
\end{split}
\end{equation}
Now we apply \eqref{vx_y0_bound} and \eqref{vx_ye_bound} to \eqref{I123_ineq1} and obtain
\begin{equation}
\begin{split}\label{I123_bound}
\left| I_1^{(t,s,\epsilon)} + I_2^{(t,s,\epsilon)}+ I_3^{(t,s,\epsilon)} \right| 
 \leq C \frac{1}{1-\epsilon},  
\end{split}
\end{equation}
where $C$ is a constant independent of $(t,s,\epsilon)$. 
Also, \thmref{asymptotic_y}, the mean value theorem, and \eqref{vx_uni_conv} imply that
\begin{align}\label{I_convergence}
\lim_{\epsilon \downarrow 0} \left( I_1^{(t,s,\epsilon)} + I_2^{(t,s,\epsilon)}+ I_3^{(t,s,\epsilon)} \right) 
= - \left|   Y_s^{(t,\hat y^0(t))}  -\hat y^0(s)  \right| \quad \textrm{almost surely.}
\end{align}

The uniform boundedness in \eqref{I123_bound} and the convergence in \eqref{I_convergence} allow us to apply the dominated convergence theorem and conclude 
\begin{align}\label{J_conv}
\lim_{\epsilon \downarrow 0 } J^\epsilon(t) = - G(t) \quad \textrm{and} \quad  \lim_{\epsilon \downarrow 0 }   \int_t^T J^\epsilon(s) ds = -\int_t^T G(s) ds,
\end{align}
where $G$ is defined in \eqref{G_def} and the function $J^\epsilon:[0,T]\to \R$ is defined as
$$
J^\epsilon(t):=\lambda \int_t^T e^{-\lambda(s-t)} \EE\left[ I_1^{(t,s,\epsilon)}+I_2^{(t,s,\epsilon)}+I_3^{(t,s,\epsilon)} \right] ds.
$$
We rewrite \eqref{psi_ode} as
\begin{align}
J^\epsilon(t) &= \psi^\epsilon(t) - \lambda  \int_t^T e^{-\lambda (s-t)}  \psi^\epsilon(s) ds=- \frac{d}{dt}\left( \int_t^T e^{-\lambda (s-t)}  \psi^\epsilon(s) ds\right). \nonumber
\end{align}
We integrate both sides above to obtain
\begin{align}\label{int_J}
 \int_t^T e^{-\lambda (s-t)}  \psi^\epsilon(s) ds= \int_t^T J^\epsilon(s) ds .
\end{align}
Using \eqref{int_J}, the equation \eqref{psi_ode} becomes 
\begin{align}\label{psi_represent2}
\psi^\epsilon(t) &=  J^\epsilon(t) + \lambda \int_t^T J^\epsilon(s) ds.
\end{align}
Finally, \eqref{J_conv} and \eqref{psi_represent2} imply \eqref{v_asymp}.
\end{proof}

\begin{corollary}
For $(t,x)\in [0,T) \times (0,1)$, 
\begin{align}
v^\epsilon(t,x)=v^0(t,x) - \left( G(t) + \lambda \int_t^T G(s) ds - \int_{\hat y^0(t)}^x F(t, \eta) d\eta    \right) \cdot \epsilon + o(\epsilon),\label{v(x)_asymp}
\end{align}
where $F$ and $G$ are defined in \eqref{F_def} and \eqref{G_def}.
\end{corollary}
\begin{proof}
The inequality \eqref{vx_diff_bound} and the limit \eqref{vxe_conv} allow us to use the dominated convergence theorem, and we obtain
\begin{align}
\lim_{\epsilon \downarrow 0} \int_{\hat y^0(t)}^x \frac{v_x^\epsilon(t,\eta)-v_x^0(t,\eta)}{\epsilon} d\eta = \int_{\hat y^0(t)}^x F(t, \eta) d\eta.\nonumber
\end{align}
Using the above limit and \thmref{asymptotic_v}, we obtain
\begin{align}
\frac{v^\epsilon(t,x)-v^0(t,x)}{\epsilon} &=\frac{v^\epsilon(t,\hat y^0(t))-v^0(t,\hat y^0(t))}{\epsilon} + \int_{\hat y^0(t)}^x \frac{v_x^\epsilon(t,\eta)-v_x^0(t,\eta)}{\epsilon} d\eta \nonumber\\
&\xrightarrow[\epsilon \downarrow 0]{}  -\left(G(t) + \lambda \int_t^T G(s) ds\right) + \int_{\hat y^0(t)}^x F(t, \eta) d\eta . \nonumber 
\end{align}
The above limit is equivalent to \eqref{v(x)_asymp}.
\end{proof}

Not surprisingly, \thmref{asymptotic_y} and \thmref{asymptotic_v} indicate that the no-trade region widens and the value function diminishes as the transaction cost parameter $\epsilon$ increases. See Figure \ref{asymp_fig} for numerical illustrations.

\medskip

 Our next task is to investigate some intertwined effects of the search frictions and transaction costs on the no-trade region and value function. To be specific, for two different search friction parameters $\lambda_1<\lambda_2$, we would like to compare the magnitude of the widening (deminishing, resp.) effects of the transaction costs on the no-trade region (the value function, resp.). Figure \ref{compare_fig} numerically describes this comparison result.
The following technical lemma turns out to be useful to the proof of this comparison.

\begin{lemma}\label{tech_asymptotic}
(i) There is a constant $C$ such that
\begin{align}\label{y0_lambda_infinity}
 \left|\hat y^{0}(t)- y_{\infty} \right| \leq \tfrac{C}{\lambda} \quad \textrm{for} \quad (t,\lambda)\in [0,T)\times [1,\infty).
\end{align}

(ii) $\frac{\partial}{\partial \lambda}v_x^0(t,x)$ and $\frac{\partial}{\partial \lambda}v_{xx}^0(t,x)$ exist for $(t,x)\in [0,T]\times (0,1)$, and
\begin{equation}
\begin{split}\label{various_limit1}
&\lim_{\lambda\to \infty} \lambda^2 \tfrac{\partial}{\partial \lambda} v_x^{0}(t,x)\Big|_{x=\hat y^{0}(t)}= 0, \quad \lim_{\lambda\to \infty} \lambda^2 \tfrac{\partial}{\partial \lambda} v_{xx}^{0}(t,x)\Big|_{x=\hat y^{0}(t)}  =  \sigma^2, \\
&\lim_{\lambda\to \infty} \lambda  v_{xx}^{0}(t,\hat y^{0}(t))= - \sigma^2,  \qquad \,\,\,\,
\lim_{\lambda\to \infty} \lambda v_{xxx}^{0}(t,\hat y^{0}(t)) =0. 
\end{split}
\end{equation}

(iii) $\frac{\partial}{\partial \lambda}\hat{y}^0(t)$ exists for $t\in [0,T)$, and there is a constant $C$ such that
\begin{equation}
\begin{split}\label{y_lambda_derivative_bound}
\left| \tfrac{\partial}{\partial \lambda}\hat{y}^0(t)\right| \leq \tfrac{C}{\lambda^2} \quad \textrm{for} \quad (t,\lambda)\in [0,T)\times [1,\infty).
\end{split}
\end{equation}

(iv) $\frac{\partial}{\partial \lambda} \EE \left[ \left\vert Y_{s}^{(t, \hat{y}^{0}(t))} - \hat{y}^{0}(s) \right\vert \right]$ exists for $0\leq t\leq s\leq T$, and there is a constant $C$ such that
\begin{equation}
\begin{split}\label{g_lambda2_bound}
\left |\tfrac{\partial}{\partial \lambda} \left( \EE \left[ \left\vert Y_{s}^{(t, \hat{y}^{0}(t))} - \hat{y}^{0}(s) \right\vert \right] \right) \right| \leq \tfrac{C}{\lambda^2}   \quad \textrm{for} \quad (t,s,\lambda)\in [0,T)\times [t,T]\times [1,\infty).
\end{split}
\end{equation}

(v) There is a constant $C$ such that 
\begin{equation}
\begin{split}\label{sqrt_g_bound}
  \tfrac{\EE \left[ \left| Y_s^{(t, x)} - x \right| \right]}{\sqrt{s - t}} \leq C \quad \textrm{for} \quad (t,s,x)\in [0,T)\times (t,T] \times (0,1),
\end{split}
\end{equation}
and for $t\in [0,T)$, 
\begin{equation}
\begin{split}\label{sqrt_g_lim}
  \lim_{s \downarrow t} \tfrac{\EE \left[ \left| Y_s^{(t, x)} - x \right| \right]}{\sqrt{s - t}} = \sigma \sqrt{\tfrac{2}{\pi}} \,\, x(1-x).
\end{split}
\end{equation}
\end{lemma}
\begin{proof}
See Appendix.
\end{proof}

\begin{figure}[t]
		\begin{center}$
			\begin{array}{cc}
			\includegraphics[width=0.45\textwidth]{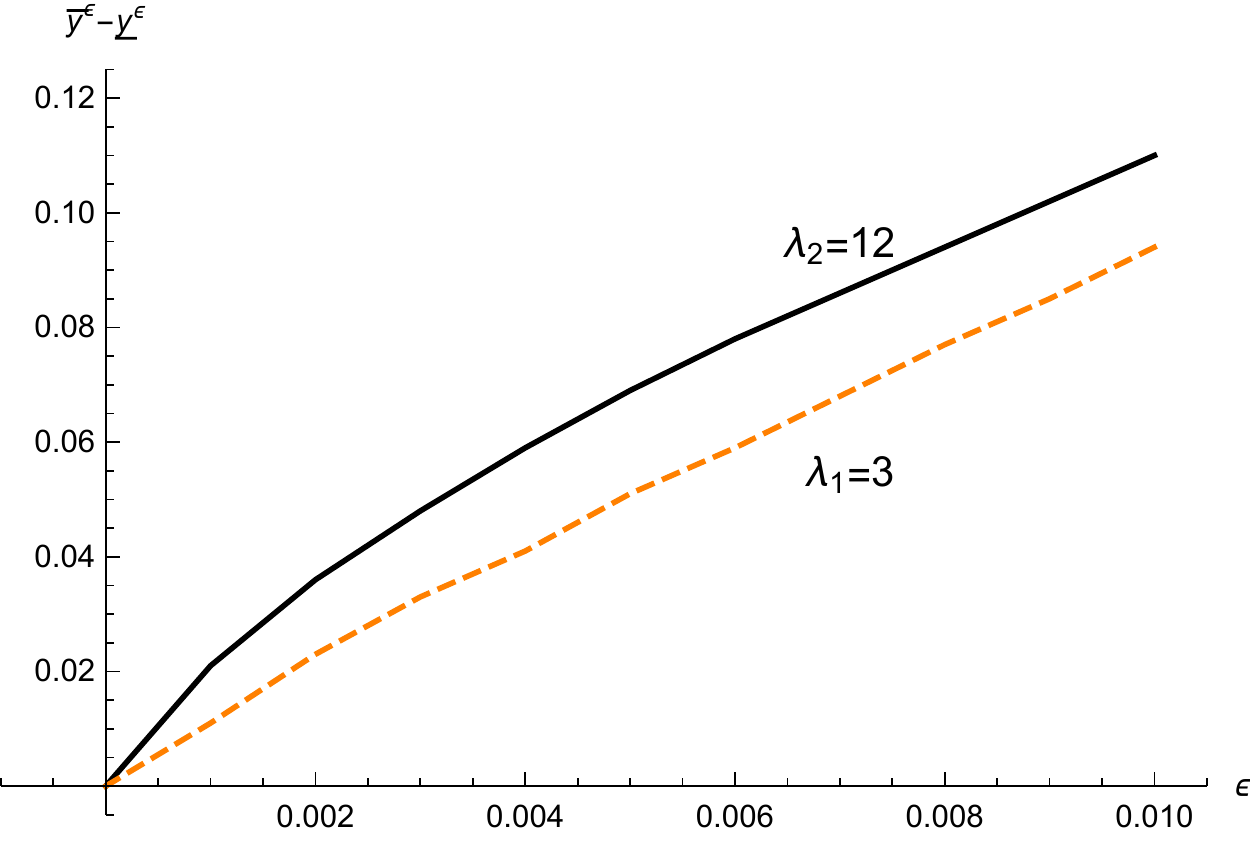} \,\,\, & \,\,\,
 			\includegraphics[width=0.45\textwidth]{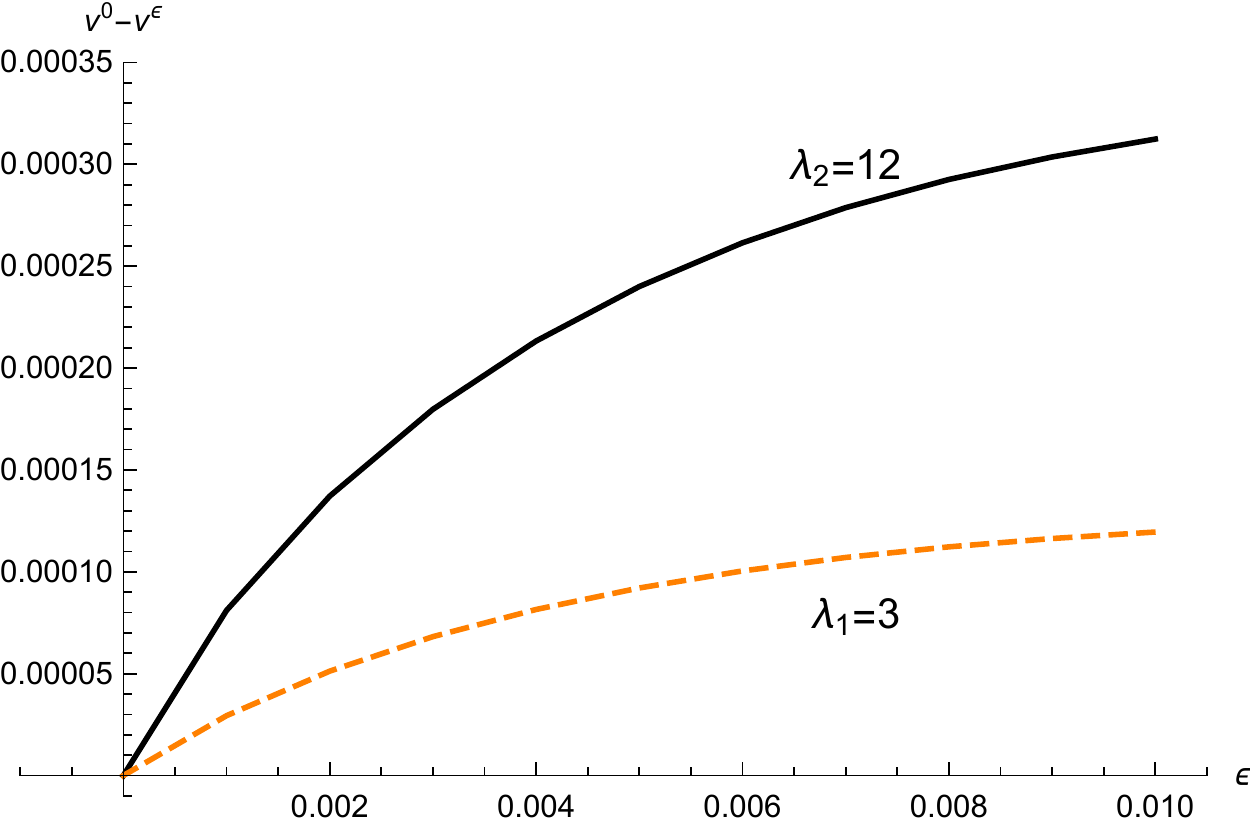} 		
			\end{array}$
	\end{center}
	 \caption{The left graph compares $\oy^\epsilon(t)-\uy^\epsilon(t)$ for $\lambda_1=3$ (dashed line) and $\lambda_2=12$ (solid line), as functions of $\epsilon$. The right graph compares $v^0(t,\hat{y}^0(t)) - v^\epsilon(t,\hat{y}^0(t))$ for $\lambda_1=3$ (dashed line) and $\lambda_2=12$ (solid line), as functions of $\epsilon$. The other parameters are $\mu=0.4, \, r=0.1, \sigma=1,\, T=1$, and $t=0.75$.
		}
		\label{compare_fig}
\end{figure}

\begin{proposition}\label{prop_comapre_frictions}
For $t \in [0,T)$, there exists $\Lambda(t)>0$ such that $\Lambda(t)<\lambda_1<\lambda_2$ implies
\begin{align}
\left\{
\begin{array}{c}
\left(\oy^\epsilon(t)-\uy^\epsilon(t)\right)\big|_{\lambda=\lambda_1} < \,\,\, \left(\oy^\epsilon(t)-\uy^\epsilon(t)\right)\big|_{\lambda=\lambda_2}\\
\left( v^{0}(t, \hat y^{0}(t)) - v^{\epsilon}(t, \hat y^{0}(t)) \right) \big|_{\lambda=\lambda_1} < \,\,\, \left( v^{0}(t, \hat y^{0}(t)) - v^{\epsilon}(t, \hat y^{0}(t)) \right) \big|_{\lambda=\lambda_2}
\end{array}
\right\}
\,\,\,\, \textrm{for small enough $\epsilon>0$.}\nonumber
\end{align}
\end{proposition}
\begin{proof}
\thmref{asymptotic_y} implies that
\begin{align}
\oy^{\epsilon} (t)-\uy^{\epsilon} (t) =  - \frac{2}{v_{xx}^0(t,\hat y^0(t))} \cdot \epsilon + o(\epsilon).\nonumber
\end{align}
Therefore, to prove the first statement, it is enough to show that $\frac{\partial}{\partial \lambda} \left( v_{xx}^0(t,\hat y^0(t)) \right)>0$ for sufficiently large $\lambda$. Indeed, using (ii) and (iii) in \lemref{tech_asymptotic}, we obtain
\begin{equation}
\begin{split}
\lambda^2 \tfrac{\partial}{\partial \lambda} \left( v_{xx}^{0}(t,\hat y^{0}(t)) \right) &= \lambda^2 \tfrac{\partial}{\partial \lambda} v_{xx}^{0}(t,x) \big|_{x=\hat y^{0}(t)} + \lambda v_{xxx}^{0}(t,\hat y^{0}(t))\cdot  \lambda\tfrac{\partial}{\partial \lambda} \hat y^{0}(t) \xrightarrow[\lambda\to \infty]{} \sigma^2>0.
\end{split}
\end{equation}

To prove the second statement, it is enough check $\frac{\partial}{\partial \lambda} \left( G(t) + \lambda \int_{t}^{T} G(s) ds \right)>0$ for sufficiently large $\lambda$, according to \thmref{asymptotic_v}. For this purpose, we focus on proving the following:
\begin{align}\label{value_asymptotic_goal}
  \lim_{\lambda \rightarrow \infty} \sqrt{\lambda} \left( \tfrac{\partial}{\partial \lambda} \left( G(t) + \lambda \int_{t}^{T} G(s) ds \right) \right)  > 0.
\end{align}
By (iv) in \lemref{tech_asymptotic}, we can take derivative inside of the expectations and obtain
\begin{align}\label{G_lambda_exp}
    \sqrt{\lambda}  \left( \tfrac{\partial}{\partial \lambda} \left( G(t) + \lambda \int_{t}^{T} G(s) ds \right) \right)  
     = \lambda^{\frac{1}{2}} \tfrac{\partial}{\partial \lambda} G(t) + \int_t^T \left( \lambda^{\frac{1}{2}}  G(s)+\lambda^{\frac{3}{2}}  \tfrac{\partial}{\partial \lambda}G(s) \right) ds,
\end{align}
where the partial derivative $\frac{\partial}{\partial \lambda} G(t)$ can be written as
\begin{align}
\tfrac{\partial}{\partial \lambda} G(t) = \int_t^T e^{-\lambda(s-t)} \left( (1-\lambda(s-t)) \EE \left[ \left\vert Y_{s}^{(t, \hat{y}^{0}(t))} - \hat{y}^{0}(s) \right\vert \right] + \lambda \tfrac{\partial}{\partial \lambda} \EE \left[ \left\vert Y_{s}^{(t, \hat{y}^{0}(t))} - \hat{y}^{0}(s) \right\vert \right]      \right) ds. \nonumber
\end{align}

Observe that (i) and (v) in \lemref{tech_asymptotic} imply  
\begin{align}\label{mixed_bound}
\EE \left[ \left\vert Y_{s}^{(t, \hat{y}^{0}(t))} - \hat{y}^{0}(s) \right\vert \right] \leq C\left(\sqrt{s-t}+ \tfrac{1}{\lambda}\right) \quad \textrm{for} \quad (t,s,\lambda)\in [0,T)\times (t,T]\times [1,\infty),
\end{align}
for a constant $C$. We use \eqref{y0_lambda_infinity}, \eqref{g_lambda2_bound}, \eqref{sqrt_g_lim}, and \eqref{lemma2_result} to obtain the following limits:
\begin{align}\label{G_limits}
\lim_{\lambda\to \infty} \lambda^{\frac{1}{2}}  G(s) = \tfrac{\sigma}{\sqrt{2}} \, y_\infty(1-y_\infty), \quad \lim_{\lambda\to \infty} \lambda^{\frac{3}{2}}\tfrac{\partial}{\partial \lambda} G(s) =- \tfrac{\sigma}{2\sqrt{2}} \, y_\infty(1-y_\infty) \quad \textrm{for} \quad s\in [0,T).
\end{align}
We apply \eqref{g_lambda2_bound}, \eqref{mixed_bound}, and \eqref{lemma2_result2} to the expressions of $G$ and $\tfrac{\partial}{\partial \lambda} G $ to obtain the boundedness:
\begin{align}\label{G_bounds}
\left| \lambda^{\frac{1}{2}}  G(s)\right| + \left|\lambda^{\frac{3}{2}}\tfrac{\partial}{\partial \lambda} G(s) \right| \leq C \quad \textrm{for} \quad (s,\lambda)\in [0,T]\times [1,\infty)
\end{align}
Finally, \eqref{G_limits} and \eqref{G_bounds} enable us to apply the dominated convergence theorem to \eqref{G_lambda_exp}:
\begin{align}
    \lim_{\lambda\to \infty}\sqrt{\lambda}  \left( \tfrac{\partial}{\partial \lambda} \left( G(t) + \lambda \int_{t}^{T} G(s) ds \right) \right)  
     =  \tfrac{\sigma}{2\sqrt{2}} \, y_\infty(1-y_\infty)(T-t)>0.
\end{align}
Therefore, we conclude \eqref{value_asymptotic_goal} and the proof is done.
\end{proof}

\begin{remark}
\proref{prop_comapre_frictions} implies that the effects of the transaction costs are more pronounced (more widening effect of the no-trade region and more diminishing effect of the value function) in the market with less search frictions. For better understanding of the result, it would be helpful to observe that the following two extreme cases ($\lambda=0$ and $\lambda=\infty$) are consistent with the result:  

(i) If we consider an extreme case of $\lambda=0$ (i.e., no trading opportunity), the transaction costs do not affect the investor anyway. In words, in the market with very severe search frictions, the effects of the transaction costs are negligible.  

(ii) If we consider an extreme case of $\lambda=\infty$ (i.e., continuous trading opportunities), it is well known that the width of the no-trade region is the order of $\epsilon^{\frac{1}{3}}$ and the decrease of the value function is the order of $\epsilon^{\frac{2}{3}}$. In contrast, when $\lambda<\infty$, \thmref{asymptotic_y} and \thmref{asymptotic_v} say that both the width of the no-trade region and the decrease of the value function are the order of $\epsilon$. Obviously, for small enough $\epsilon$, we have $\epsilon<\epsilon^{\frac{1}{3}}$ and $\epsilon<\epsilon^{\frac{2}{3}}$. This implies that in the market with the search frictions ($\lambda<\infty$), the effects of the transaction costs are less pronounces, compared to the market with no search friction ($\lambda=\infty$).
\end{remark}

\section{Conclusion}
This paper presents an utility maximization problem of the terminal wealth, in a market with two different types of the illiquidity: search frictions and transaction costs. We show the existence of the solution to the HJB equation that is regular enough, and provide the verification argument. The optimal trading strategy is characterized by a no-trade region, where the boundary points of the no-trade region are uniquely determined by the strict concavity of the value function. In Proposition 4.4, we provide some conditions under which the investor wants to achieve zero/positive stock/bond holdings.
We provide the asymptotic expansions of the boundaries of the no-trade region and the value function for small transaction costs. We further show that reduction of the search frictions amplifies the effects of the transaction costs (more widening effect of the no-trade region and more diminishing effect of the value function).

As a future research, we plan to investigate a joint limiting behavior of the no-trade region and the value function for small transaction costs and large arrival rate. For that purpose, we think it could be useful to 
consider some stationary versions of the problem with a hope to obtain explicit formulas independent of the time variable.

\section*{Acknowledgement}
This work was supported by the National Research Foundation of Korea (NRF) grant funded by the Korea government (MSIT) (No. 2020R1C1C1A01014142 and No. 2021R1I1A1A01050679).

\bibliographystyle{siam}  
\bibliography{reference}        


\appendix

\section{Proof of \lemref{PDE_existence}}

Since the parabolic type PDE \eqref{hjb_v} is not uniformly elliptic, we change variable as $x=h(z):=\frac{e^z}{1+e^{z}}$ and consider the PDE for $v(t,h(z))$.
 
 To handle the nonlinear term in the PDE, we first consider the following map $\phi$ from $C_b([0,T]\times \R)$ (equipped with the uniform norm) to itself:
\begin{align}\label{phi_def}
\phi(f)(t,z):=\int_t^T e^{-\lambda(s-t)} \EE \left[ K_f(s,Z_s^{(t,z)}) \right] ds,
\end{align}
where for $(s,z)\in [t,T]\times \R$,
\begin{equation}
\begin{split}
Z_s^{(t,z)}&:=z+ \left(\mu-r-\tfrac{\sigma^2}{2} \right)(s-t) + \sigma(B_s-B_t),\\
K_f(s,z)&:=  (\mu-r)h(z)+r-\tfrac{1}{2}\sigma^2 h(z)^2 \\
& \qquad + \lambda   \sup_{\zeta\in \R}  \left( f(s,\zeta)-  \ln \left(\tfrac{1+\be h(\zeta)}{1+\be h(z)} \right) 1_{\{z<\zeta\}} -\ln \left(\tfrac{1-\ue h(\zeta)}{1-\ue h(z)} \right) 1_{\{z>\zeta\}}    \right).
\end{split}
\end{equation} 
We observe that for $f,g \in C_b([0,T]\times \R)$, 
\begin{equation}
\begin{split}
\| \phi(f)-\phi(g) \|_\infty & \leq \lambda \int_t^T e^{-\lambda(s-t)} \left( \sup_{\zeta \in \R} \left| f(s,\zeta)-g(s,\zeta)  \right|  \right) ds\\
& \leq \left(1- e^{-\lambda T} \right) \| f-g \|_\infty,
\end{split}
\end{equation} 
where the first inequality is due to the triangular inequality for supremum. Therefore, the map $\phi$ in \eqref{phi_def} is a contraction map and there exists a unique $\hat u\in C_b([0,T]\times \R)$ such that $\phi(\hat u)=\hat u$, by the Banach fixed point theorem. 

\medskip

{\bf Claim}: For all $\delta\in (0,1)$, $K_{\hat{u}}\in C^{\frac{\delta}{2},\delta}([0,T]\times \R)$.

(Proof of Claim):
We first check that $K_{\hat u}$ is bounded. Since $0<h(z)<1$, we observe that 
\begin{align}\label{K_bound}
|K_{\hat u}(t,z) | \leq   |\mu|+2|r|  +\tfrac{\sigma^2}{2}+ \lambda \left( \| \hat u \|_\infty + \ln \left(\tfrac{1+\be}{1-\ue} \right) \right),
\end{align}
i.e., $\| K_{\hat u}\|_\infty<\infty$. Therefore, to prove the claim, it is enough to check that $K_{\hat u}$ is uniformly Lipschitz with respect to $z$ variable and $\tfrac{1}{2}$-H\"older continuous with respect to $t$ variable.
For $\Delta>0$, we obtain the following inequalities:
\begin{equation}
\begin{split}\label{K_Lip}
&|K_{\hat u}(t,z+\Delta)- K_{\hat u}(t,z) | \\
&\leq   (|\mu-r|+\sigma^2)\Delta + \lambda \sup_{\zeta\in \R} \left| 1_{\{z+\Delta<\zeta\}}\cdot  \ln\left(\tfrac{1+\be h(z+\Delta)}{1+\be h(z)}\right) + 1_{\{z<\zeta \leq z+\Delta \}}  \cdot \ln\left(\tfrac{1+\be h(\zeta)}{1+\be h(z)} \right) \right| \\
& \quad +  \lambda \sup_{\zeta\in \R} \left| 1_{\{z+\Delta>\zeta\}} \cdot \ln\left(\tfrac{1-\ue h(z+\Delta)}{1-\ue h(z)}\right) - 1_{\{z\leq \zeta < z+\Delta \}} \cdot \ln\left(\tfrac{1-\ue h(\zeta)}{1-\ue h(z)} \right) \right|\\
& \leq  \left(|\mu-r|  +\sigma^2 + \lambda \left( \be + \tfrac{\ue}{1-\ue}   \right) \right) \Delta,
\end{split}
\end{equation} 
where we used the mean value theorem and the bounds $0<h<1$ and $0<h'<1$. We can treat $\Delta<0$ by the same way, and conclude that $K_{\hat{u}}$ is uniformly Lipschitz with respect to $z$. 

For $\Delta>0$, using \eqref{K_Lip} and the mean value theorem, we observe that
\begin{equation}
\begin{split}\label{K_delta}
&\EE\left| e^{\lambda \Delta} K_{\hat u}(s,Z_s^{(t+\Delta,z)})-K_{\hat u}(s,Z_s^{(t,z)})  \right|  \\
&\leq \lambda e^{\lambda T} \| K_{\hat u}\|_\infty \Delta +   \EE\left| K_{\hat u}(s,Z_s^{(t+\Delta,z)})-K_{\hat u}(s,Z_s^{(t,z)})  \right|   \\
& \leq \lambda e^{\lambda T}\| K_{\hat u}\|_\infty \Delta +   \left(|\mu-r| +\sigma^2 + \lambda \left( \be + \tfrac{\ue}{1-\ue}   \right) \right)  \EE \left|  Z_s^{(t+\Delta,z)}-Z_s^{(t,z)}  \right| \\
&\leq \lambda e^{\lambda T}\| K_{\hat u}\|_\infty \Delta +   \left(|\mu-r| +\sigma^2 + \lambda \left( \be + \tfrac{\ue}{1-\ue}   \right) \right) \left( \left|\mu-r-\tfrac{\sigma^2}{2}\right|\Delta  + \sigma \sqrt{\tfrac{2}{\pi}} \, \Delta^{\frac{1}{2}}    \right),
\end{split}
\end{equation} 
where the last inequality is due to 
$$
\EE \left|  Z_s^{(t+\Delta,z)}-Z_s^{(t,z)}  \right| \leq  \left|\mu-r-\tfrac{\sigma^2}{2}\right|\Delta  + \sigma\, \EE \left| B_{t+\Delta}-B_t \right|
 = \left|\mu-r-\tfrac{\sigma^2}{2}\right|\Delta  + \sigma \sqrt{\tfrac{2}{\pi}} \, \Delta^{\frac{1}{2}}.
$$
Using $\hat u= \phi(\hat u)$ and triangular inequality, we obtain
\begin{displaymath}
\begin{split}
&|K_{\hat u}(t+\Delta,z)- K_{\hat u}(t,z) |\leq \lambda \sup_{\zeta \in R} \left|  \hat u(t+\Delta,\zeta)-\hat u(t,\zeta) \right |\\
&=\lambda \sup_{\zeta\in \R} \left | \int_{t+\Delta}^T e^{-\lambda (s-t-\Delta)}  \EE\left[  K_{\hat u}(s,Z_s^{(t+ \Delta,\zeta)}) \right] ds -\int_{t}^T e^{-\lambda (s-t)}  \EE \left[ K_{\hat u}(s,Z_s^{(t,\zeta)}) \right] ds \right|\\
&\leq \lambda \sup_{\zeta\in \R} \bigg( \int_{t+\Delta}^T e^{-\lambda (s-t)} \EE\left| e^{\lambda \Delta} K_{\hat u}(s,Z_s^{(t+ \Delta,\zeta)})-K_{\hat u}(s,Z_s^{(t,\zeta)}) \right| ds  \\
&\qquad \qquad\qquad\qquad\qquad\qquad \qquad +\int_t^{t+\Delta} e^{-\lambda (s-t)} \EE \left| K_{\hat u}(s,Z_s^{(t,\zeta)}) \right| ds \bigg) \\
&\leq C \Delta^{\frac{1}{2}},
\end{split}
\end{displaymath} 
where the generic constant $C$ only depends on the market parameters and $\|\hat u\|_{\infty}$ and the last inequality is due to \eqref{K_bound} and \eqref{K_delta}. We can treat $\Delta<0$ by the same way, and conclude that $K_{\hat u}$ is $\tfrac{1}{2}$-H\"older continuous with respect to $t$ variable. \\
(End of the proof of Claim).

\medskip

Let $\delta\in (0,1)$ be fixed. Then, the above claim and Theorem 9.2.3 in \cite{krylov1996lectures} ensure that there exists a unique function $u\in C^{1+\frac{\delta}{2},2+\delta}((0,T)\times \R)$ such that it satisfies the following PDE on $(t,z)\in (0,T)\times \R$.
\begin{equation}\label{u_pde}
\begin{split}
\begin{cases}
0=u(T,z),\\
0= u_t  + (\mu-r-\frac{\sigma^2}{2} )u_z+\frac{1}{2} \sigma^2 u_{zz} - \lambda u + K_{\hat u}
\end{cases}
\end{split}
\end{equation} 
Since $u$ admits a unique continuous extension on $[0,T]\times \R$ (i.e., see chapter 8.5 in \cite{krylov1996lectures}), we let $u\in C^{1+\frac{\delta}{2},2+\delta}([0,T]\times \R)$. By the Feynman-Kac formula (i.e., see Theorem 5.7.6 in \cite{karatzas2014brownian}), the solution $u$ of the parabolic PDE \eqref{u_pde} has the stochastic representation $u=\phi(\hat u)$, where $\phi$ is defined in \eqref{phi_def}. Since $\hat u$ is chosen as the unique fixed point of the map $\phi$, we conclude that $u=\hat u$.

Our next task is to define $u(t,\pm \infty)$ for $t\in [0,T]$. Using $\lim_{z\to \infty}h(Z_s^{(t,z)})=1$ and $\lim_{z\to -\infty}h(Z_s^{(t,z)})=0$ almost surely, we obtain
\begin{equation}\label{K_limit}
\begin{split}
\begin{cases}
\lim_{z\to \infty} K_u(s,Z_s^{(t,z)}) = \mu-\tfrac{1}{2}\sigma^2 + \lambda   \sup_{\zeta\in \R}  \left( u(s,\zeta) -\ln \left(\frac{1-\ue h(\zeta)}{1-\ue } \right)   \right) \\
\lim_{z\to -\infty} K_u(s,Z_s^{(t,z)}) = r+ \lambda   \sup_{\zeta\in \R}  \left( u(s,\zeta)-  \ln \left(1+\be h(\zeta) \right)  \right) 
\end{cases}  a.s.
\end{split}
\end{equation} 
The above convergence and $\| K_u\|_\infty<\infty$ enable us to apply the dominated convergence theorem: 
\begin{displaymath}
\begin{split}
\lim_{z\to \pm \infty} u(t,z) 
&=\lim_{z\to \pm\infty}  \phi(u)(t,z)=\lim_{z\to \pm\infty} \int_t^T e^{-\lambda(s-t)} \EE \left[ K_u(s,Z_s^{(t,z)}) \right] ds\\
&=
\begin{cases}
\int_t^T e^{-\lambda(s-t)}  \left(  \mu-\tfrac{1}{2}\sigma^2 + \lambda   \sup_{\zeta\in \R}  \left( u(s,\zeta) -\ln \left(\frac{1-\ue h(\zeta)}{1-\ue } \right)   \right)   \right) ds, &\textrm{for $z\to \infty$}\\
\int_t^T e^{-\lambda(s-t)}  \left(   r+ \lambda   \sup_{\zeta\in \R}  \left( u(s,\zeta)-  \ln \left(1+\be h(\zeta) \right)  \right)   \right) ds, &\textrm{for $z\to -\infty$}\
\end{cases}.
\end{split}
\end{displaymath} 
Therefore, we can continuously extend $u$ to $z=\pm \infty$ and $u(t,\infty)$ and $u(t,-\infty)$ are defined by the above limit. We observe that for $z\in \{ \infty , -\infty\}$, $u(t,z)$ satisfies
\begin{equation}\label{u01_pde}
\begin{split}
0&=u_t(t,z) + (\mu-r)h(z)+r-\tfrac{1}{2}\sigma^2 h(z)^2 - \lambda u(t,z) \\
&\qquad+ \lambda   \sup_{\zeta\in \R}  \left( u(t,\zeta)-  \ln \left(\tfrac{1+\be h(\zeta)}{1+\be h(z)} \right) 1_{\{z<\zeta\}} -\ln \left(\tfrac{1-\ue h(\zeta)}{1-\ue h(z)} \right) 1_{\{z>\zeta\}}    \right),
\end{split}
\end{equation}
where the function $h$ is continuously extended as $h(\infty):=1$ and $h(-\infty):=0$.

Now we define $v$ as $v(t,x):=u(t,h^{-1}(x))$ for $(t,x)\in [0,T]\times [0,1]$. Such $v$ is well-defined because $h:\R\cup \{ \infty, -\infty\} \to [0,1]$ is bijective. We observe that for $(t,x)\in (0,T)\times (0,1)$ and $z=h^{-1}(x)$,
\begin{equation}\label{change_variable}
\begin{split}
v_t(t,x)&=u_t(t,z),\\
x(1-x)v_x(t,x)&=u_z(t,z),\\
x^2(1-x)^2 v_{xx}(t,x)&=u_{zz}(t,z)-(1-2x)u_z(t,z).
\end{split}
\end{equation}

The PDE for $u$ in \eqref{u_pde} with $\hat u$ replaced by $u$ and the equalities in \eqref{change_variable} produce the PDE for $v$, which is \eqref{hjb_v}. Therefore, statement (i) is valid. 

To check statement (ii), we observe that $v(t,0)=u(t,-\infty)$ and $v(t,1)=u(t,\infty)$. Then, the continuous differentiability of $v(t,0)$ and $v(t,1)$ with respect to $t$ is followed by that of $u(t,-\infty)$ and $u(t,\infty)$, and \eqref{u01_pde} produces \eqref{hjb_v01}.

Finally, statement (iii) is a direct consequence of \eqref{change_variable} and $u\in C^{1+\frac{\delta}{2},2+\delta}([0,T]\times \R)$.

\section{Proof of \lemref{tech_asymptotitcs}}

(i) When $\epsilon=0$, \lemref{uyoy_lem} and \eqref{y0} produce \eqref{vx0_bound}. 

To prove \eqref{gxx>0_all}, we first check that $Y_s^{(t,x)}$ in \eqref{YL_def} safisfies 
\begin{align}\label{Y_SDE}
dY_s^{(t,x)}=Y_s^{(t,x)} (1-Y_s^{(t,x)}) \left( (\mu-r-\sigma^2 Y_s^{(t,x)})ds + \sigma dB_s  \right).
\end{align}
Then application of Ito's formula produces that for $(s,x)\in [t,T) \times (0,1)$, 
\begin{align}
\left(Y_s^{(t,x)}-x\right)^2&= \int_t^s  Y_u^{(t,x)} (1-Y_u^{(t,x)})\left( 2(Y_u^{(t,x)}-x )( \mu-r-\sigma^2 Y_u^{(t,x)}) + \sigma^2 Y_u^{(t,x)} (1-Y_u^{(t,x)}) \right) du \nonumber \\
&\qquad +  \int_t^s 2\sigma (Y_u^{(t,x)}-x ) Y_u^{(t,x)} (1-Y_u^{(t,x)})  dB_u. \nonumber
\end{align}
Since $0<Y_s^{(t,x)}<1$, the local martingale part above is a true martingale and we obtain
\begin{align}
&\tfrac{\partial}{\partial s} \left( \EE\left[ \left(Y_s^{(t,x)}-x\right)^2 \right] \right) \nonumber \\
&=\EE \left[ Y_s^{(t,x)} (1-Y_s^{(t,x)})\left( 2(Y_s^{(t,x)}-x )( \mu-r-\sigma^2 Y_s^{(t,x)}) + \sigma^2 Y_s^{(t,x)} (1-Y_s^{(t,x)}) \right)\right]    \nonumber \\
&=-x^2(1-x)^2 \cdot \EE \left[ \tfrac{\partial}{\partial x}\left( \tfrac{Y_s^{(t,x)} \left(1-Y_s^{(t,x)}\right) \left( \mu-r-\sigma^2 Y_s^{(t,x)}\right)}{x(1-x)} \right)  \right], \label{derivative_change}
\end{align}
where the second equality is from elementary computations using the definition of $Y_s^{(t,x)}$ in \eqref{YL_def}.

For $x\in (0,1)$, we observe that 
\begin{align}\label{yyy}
\tfrac{Y_s^{(t,x)}-x}{x(1-x)} = \tfrac{A^{(t,s)}-1}{A^{(t,s)} x+ (1-x)} \quad \textrm{with} \quad A^{(t,s)}:=e^{(\mu-r-\frac{\sigma^2}{2})(s-t)+\sigma (B_s-B_t)},
\end{align}
and the above expression is decreasing in $x$, therefore,
\begin{align}\label{yyy_bound}
1- \tfrac{1}{A^{(t,s)}}<\tfrac{Y_s^{(t,x)}-x}{x(1-x)}<A^{(t,s)}-1 \quad \textrm{for} \quad 0<x<1.
\end{align}

When $\epsilon=0$, the representation of $v_x$ in \eqref{vx_exp} becomes
\begin{align}\label{vx0_exp}
v_x^0(t,x)=\int_t^T e^{-\lambda (s-t)} \EE \left[\tfrac{Y_s^{(t,x)}\left(1-Y_s^{(t,x)}\right)\left( \mu-r-\sigma^2 Y_s^{(t,x)}  \right)}{x(1-x)}  \right] ds,
\end{align}
We take derivative with respect to $x$ in the above expression. Then, the mean value theorem and the dominated convergence theorem, together with the inequalities \eqref{YD_bound} and \eqref{yyy_bound}, allow us to take derivative inside of the expectation and obtain that for $(t,x)\in [0,T)\times (0,1)$,
\begin{align}
v_{xx}^0(t,x)&=\int_t^T e^{-\lambda (s-t)} \EE \left[ \tfrac{\partial}{\partial x}\left( \tfrac{Y_s^{(t,x)} \left(1-Y_s^{(t,x)}\right) \left( \mu-r-\sigma^2 Y_s^{(t,x)}\right)}{x(1-x)} \right)  \right] ds \label{vxx0_rep}  \\
& =- \int_t^T e^{-\lambda (s-t)}  \tfrac{\partial}{\partial s} \left( \EE\left[ \left(\tfrac{Y_s^{(t,x)}-x}{x(1-x)}\right)^2 \right] \right) \, ds  \nonumber \\
&=- e^{-\lambda(T-t)} \, \EE \left[ \left(\tfrac{Y_T^{(t,x)}-x}{x(1-x)}\right)^2 \right] - \lambda \int_t^T e^{-\lambda (s-t)}  \, \EE\left[ \left(\tfrac{Y_s^{(t,x)}-x}{x(1-x)}\right)^2 \right]  \, ds,  \label{vxx0_exp}
\end{align}
where the second equality is due to \eqref{derivative_change}, and the third equality is from  integration by parts. Obviously \eqref{vxx0_exp} implies that $v_{xx}^0(t,x)<0$ for $(t,x)\in [0,T)\times (0,1)$. 

To conclude \eqref{gxx>0_all}, it only remains to check that $\lim_{x\uparrow 1} v_{xx}^0(t,x)<0$ and $\lim_{x\downarrow 0} v_{xx}^0(t,x)<0$. 
Indeed, \eqref{yyy} and \eqref{yyy_bound} enable us to apply the dominated convergence theorem to \eqref{vxx0_exp} and obtain 
\begin{align}
\lim_{x\uparrow 1} v_{xx}^0(t,x)&= - e^{-\lambda(T-t)} \, \EE \Big[ \left(1- \tfrac{1}{A^{(t,T)}}\right)^2 \Big] - \lambda \int_t^T e^{-\lambda (s-t)}  \, \EE \Big[ \left(1- \tfrac{1}{A^{(t,s)}}\right)^2 \Big] ds ,\nonumber\\
\lim_{x\downarrow 0} v_{xx}^0(t,x)&= - e^{-\lambda(T-t)} \, \EE \Big[ \left(A^{(t,T)}-1\right)^2 \Big] - \lambda \int_t^T e^{-\lambda (s-t)}  \, \EE \Big[  \left(A^{(t,s)}-1\right)^2 \Big] ds,\nonumber
\end{align}
and we conclude that $\lim_{x\uparrow 1} v_{xx}^0(t,x)<0$ and $\lim_{x\downarrow 0} v_{xx}^0(t,x)<0$.

\medskip

(ii) The SDE for $Y_s^{(t,x)}$ in \eqref{Y_SDE} and $0<Y_s^{(t,x)}<1$ imply that for $(s,x)\in [t,T) \times (0,1)$, 
\begin{equation}\label{EYs_exp}
\tfrac{\partial}{\partial s} \left(  \EE\left[ Y_s^{(t,x)} \right] \right) = \EE \left[ Y_s^{(t,x)} \left( 1- Y_s^{(t,x)}\right) \left( \mu - r - \sigma^2 Y_s^{(t,x)}  \right)   \right]
\end{equation}
We divide both sides of \eqref{EYs_exp} by $x(1-x)$ and take derivative with respect to $x$. Then, we can put the derivative inside of the expectation as in the proof of part (i), and obtain
\begin{align}
 \EE \left[ \tfrac{\partial}{\partial x} \left( \tfrac{Y_s^{(t,x)} \left( 1- Y_s^{(t,x)}\right) \left( \mu - r - \sigma^2 Y_s^{(t,x)}  \right) }{x(1-x)} \right)  \right]&=\tfrac{\partial}{\partial s} \left(  \EE\left[\tfrac{\partial}{\partial x} \left(\tfrac{Y_s^{(t,x)} }{x(1-x)}\right) \right] \right) \nonumber\\
 &=\tfrac{\partial}{\partial s} \left(  \tfrac{  \EE\left[\frac{\partial}{\partial x}Y_s^{(t,x)} \right] }{x(1-x)} - \tfrac{(1-2x)\EE\left[Y_s^{(t,x)} \right]}{x^2(1-x)^2} \right).\nonumber
\end{align}
We rearrange the above equation and use \eqref{EYs_exp}, \eqref{vx0_exp}, and \eqref{vxx0_exp} to obtain
\begin{align}\label{F_conversion}
\int_t^T e^{-\lambda(s-t)} \tfrac{\partial}{\partial s} \left( \EE \left[ \tfrac{\partial}{\partial x} Y_s^{(t,x)} \right]\right) ds = x(1-x)v_{xx}^0(t,x) + (1-2x) v_x^0(t,x).
\end{align}
Now we conclude that $F(t, \hat y^0(t))<1$ by the following way:
\begin{align}
F(t, \hat y^0(t))&\leq \lambda \int_t^T  e^{-\lambda(s-t)} \EE \left[ \tfrac{\partial}{\partial x} Y_s^{(t,x)} \right] ds \, \Big|_{x=\hat y^0(t)} \nonumber\\
&=\left(1- e^{-\lambda(T-t)} \EE \left[ \tfrac{\partial}{\partial x} Y_T^{(t,x)} \right] + x(1-x)v_{xx}^0(t,x) + (1-2x) v_x^0(t,x) \right) \Big|_{x=\hat y^0(t)} \nonumber\\ 
&<1, \nonumber
\end{align}
where the first inequality is from the definition of $F$ in \eqref{F_def} and the positivity of $\frac{\partial}{\partial x} Y_s^{(t,x)}$ (see \eqref{YD_bound}), and the equality is due to integration by parts and \eqref{F_conversion}, and the last inequality is due to the positivity of $\frac{\partial}{\partial x} Y_T^{(t,x)}$, $v_x^0(t,\hat y^0(t))=0$, and $v_{xx}^0(t,\hat y^0(t))<0$.

We can check that $F(t, \hat y^0(t))>-1$ by the same way as above.

\medskip

(iii) The expression in \eqref{vx_exp} produces  
\begin{equation}
\begin{split}\label{vx-v0_conv}
v_{x}^\epsilon(t,x_\epsilon)-v_x^0(t,x_\epsilon)   &=    \lambda  \int_t^T e^{-\lambda (s-t)} \EE \left[ \left( \tfrac{\partial}{\partial x} Y_s^{(t,x_\epsilon)}  \right) L_y^\epsilon(s, Y_s^{(t,x_\epsilon)} )  \right] ds.
\end{split}
\end{equation}
In the above expression, when we take limit as $\epsilon\downarrow 0$, the inequalities \eqref{Lx_bound} and \eqref{YD_bound} enable us to use the dominated convergence theorem to conclude that 
$$\lim_{\epsilon\downarrow 0} \left( v_{x}^\epsilon(t,x_\epsilon)-v_x^0(t,x_\epsilon) \right)=0.$$
The above limit and the continuity of $v_{x}^0$ implies \eqref{vx_uni_conv}.

To prove \eqref{vxx_uni_conv}, we first observe that for $(t,x)\in [0,T)\times (0,1)$ and $(s,z)\in (t,T)\times (0,1)$, the density function for $Y_s^{(t,x)}$ is given by
\begin{equation}
\begin{split}\label{Y_density}
\varphi(s,z;t,x)&:= \tfrac{\partial}{\partial z} \PP\left( Y_s^{(t,x)} \leq z \right) =\tfrac{\exp \left( -\frac{1}{2\sigma^2 (s-t)} \left( (r-\mu+\frac{\sigma^2}{2})(s-t) + \ln \left(\frac{z(1-x)}{(1-z)x} \right) \right)^2  \right)}{\sigma z(1-z)\sqrt{2\pi (s-t)}}.
\end{split}
\end{equation}
Then, the expression in \eqref{vx_exp} and $\frac{\partial}{\partial x} Y_s^{(t,x)}=\frac{Y_s^{(t,x)}(1-Y_s^{(t,x)})}{x(1-x)}$ imply that
\begin{equation}
\begin{split}\label{vx-v0}
v_x^\epsilon(t,x)-v_x^0(t,x)   &=    \lambda  \int_t^T e^{-\lambda (s-t)} \EE \left[  \tfrac{Y_s^{(t,x)}\left(1-Y_s^{(t,x)}\right)}{x(1-x)} L_y^\epsilon(s, Y_s^{(t,x)} )  \right] ds\\
&=\lambda  \int_t^T e^{-\lambda (s-t)} \left( \int_0^1 \tfrac{z(1-z)}{x(1-x)} L_y^\epsilon(s, z) \varphi(s,z;t,x) dz  \right) ds.
\end{split}
\end{equation}
For $x\in (0,1)$, direct computations produce
\begin{equation}
\begin{split}\label{density_derivative}
\tfrac{\partial}{\partial x} \left ( 
\tfrac{z(1-z)}{x(1-x)} \varphi(s,z;t,x)
\right)
= 
\tfrac{(1-z)z \left( r-\mu + (2x-\frac{1}{2})\sigma^2 + \frac{1}{s-t} \ln \left(\frac{z(1-x)}{(1-z)x} \right)  \right) \varphi(s,z;t,x) }{(1-x)^2 x^2 \sigma^2} .
 \end{split}
\end{equation}
 \assref{assumption} implies that 
\begin{align}\label{para_ineq}
\left| r-\mu + \tfrac{\sigma^2}{2}  \right| \leq \sigma^2, \quad \left| r-\mu + (2x-\tfrac{1}{2})\sigma^2  \right| \leq 2\sigma^2 \quad  \textrm{for} \quad x\in (0,1).
\end{align} 
 
Then, \eqref{Y_density} and \eqref{density_derivative} produce the following:
\begin{equation}
\begin{split}\label{nasty1}
& \int_0^1\left| \frac{\partial}{\partial x} \left ( 
\tfrac{z(1-z)}{x(1-x)} \varphi(s,z;t,x) 
\right) \right| \, dz\\
&\leq  
\int_0^1\tfrac{\left( 2\sigma^2 + \frac{1}{s-t} \left| \ln \left(\tfrac{z(1-x)}{(1-z)x} \right) \right|  \right)}{\sigma^3(1-x)^2 x^2  \sqrt{2\pi (s-t)}}\exp \left( -\tfrac{\left(\ln \left(\frac{z(1-x)}{(1-z)x} \right) \right)^2}{2\sigma^2 (s-t)} + \left|\ln \left(\tfrac{z(1-x)}{(1-z)x} \right) \right|\right) \, dz \\
&= 
 \int_{-\infty}^\infty\tfrac{\left( 2\sigma^2 + \frac{1}{s-t} \left| \zeta \right|  \right) e^{ -\frac{\zeta^2}{2\sigma^2 (s-t)} + \left|\zeta  \right|}}{\sigma^3(1-x)^2 x^2  \sqrt{2\pi (s-t)}} \cdot \tfrac{\frac{x}{1-x}e^\zeta}{(1+\frac{x}{1-x} e^\zeta)^2}  \, d\zeta \\
&\leq
 \int_{-\infty}^\infty\tfrac{\left( 2\sigma^2 + \frac{1}{s-t} \left| \zeta \right|  \right) e^{ -\frac{\zeta^2}{4\sigma^2 (s-t)} + 4\sigma^2(s-t)}}{\sigma^3(1-x)^3 x  \sqrt{2\pi (s-t)}}  \, d\zeta \\
&=\tfrac{2\sqrt{2} }{(1-x)^3 x}\left(  1 + \tfrac{1}{\sigma\sqrt{\pi (s-t)}} \right)e^{4\sigma^2(s-t)},
 \end{split}
\end{equation}
where the first inequality is due to \eqref{para_ineq}, and the first equality is obtained by change of variables as $\zeta=\ln \left(\frac{z(1-x)}{(1-z)x} \right)$. The second inequality is due to the inequality of arithmetic and geometric means, and the second equality is obtained by direct computations. 

By \eqref{nasty1}, we conclude that 
\begin{align}\label{integrable1}
\int_t^T e^{-\lambda(s-t) }\int_0^1\left| \tfrac{\partial}{\partial x} \left ( 
\tfrac{z(1-z)}{x(1-x)} \varphi(s,z;t,x) 
\right) \right| \, dz \, ds <\infty,
\end{align}
and the function $H:[0,T)\times (0,1) \to \R $ given by
\begin{align}
H(t,x):=\lambda \int_t^T e^{-\lambda(s-t) }\int_0^1 \tfrac{\partial}{\partial x} \left ( 
\tfrac{z(1-z)}{x(1-x)} \varphi(s,z;t,x) 
\right) L_y^\epsilon(s,z) \, dz \, ds 
\end{align}
is well-defined due to \eqref{integrable1} and the boundedness $| L_y^\epsilon |\leq \frac{\epsilon}{1-\epsilon}$ in \eqref{Lx_bound}. Then, \eqref{nasty1} implies that
\begin{align}\label{H_bound}
\left| H(t,x)\right| \leq \tfrac{2\sqrt{2}\lambda e^{4\sigma^2 T}}{(1-x)^3 x} \left(T + \tfrac{2\sqrt{T}}{\sigma \sqrt{\pi}} \right) \cdot \tfrac{\epsilon}{1-\epsilon} \quad \textrm{for} \quad (t,x)\in [0,T)\times (0,1).
\end{align}

Now, let's check that 
\begin{align}\label{H_vxx}
H(t,x)=v_{xx}^\epsilon(t,x)-v_{xx}^0(t,x).
\end{align}
Indeed, for $(t,x)\in [0,T)\times (0,1)$,
\begin{align}
H(t,x)&=\lim_{\delta\to 0} \frac{1}{\delta}\int_x^{x+\delta}H(t,\eta)d\eta \nonumber \\
&=\lim_{\delta\to 0}      \frac{1}{\delta} \,\, \lambda \int_t^T e^{-\lambda(s-t)} \int_0^1 \left(\tfrac{z(1-z)\varphi(s,z;t,x+\delta) }{(x+\delta)(1-x-\delta)} - \tfrac{z(1-z)\varphi(s,z;t,x) }{x(1-x)}  \right) L_y^\epsilon(s,z)   \, dz \, ds \nonumber\\
&=\lim_{\delta\to 0}      \frac{1}{\delta} \left( \left(v_x^\epsilon(t,x+\delta)-v_x^0(t,x+\delta)\right)-\left(v_x^\epsilon(t,x)-v_x^0(t,x)\right)\right)  \nonumber\\
&=v_{xx}^\epsilon(t,x)-v_{xx}^0(t,x), \nonumber
\end{align}
where the second equality is due to Fubini's theorem and the fundamental theorem of calculus, and the third equality is from \eqref{vx-v0}. 

Finally, we conclude \eqref{vxx_uni_conv} by the following observation:
\begin{align}
\limsup_{\epsilon \downarrow 0} \left| v_{xx}^\epsilon(t,x_\epsilon)-v_{xx}^0(t,x_0)\right| &\leq \limsup_{\epsilon \downarrow 0} \left| v_{xx}^\epsilon(t,x_\epsilon)-v_{xx}^0(t,x_\epsilon)\right|+ \limsup_{\epsilon \downarrow 0} \left| v_{xx}^0(t,x_\epsilon)-v_{xx}^0(t,x_0)\right| \nonumber\\
&\leq \limsup_{\epsilon \downarrow 0}\tfrac{2\sqrt{2}\lambda e^{4\sigma^2 T}}{(1-x_\epsilon)^3 x_\epsilon} \left(T + \tfrac{2\sqrt{T}}{\sigma \sqrt{\pi}} \right) \cdot \tfrac{\epsilon}{1-\epsilon} =0, \nonumber
\end{align}
where the second inequality is due to \eqref{H_bound}, \eqref{H_vxx}, and the continuity of $v_{xx}^0$.

\section{Proof of \lemref{tech_asymptotic}}

(i) This result holds due to equation (2.7) in \cite{matsumoto2006} and the inequality $|\hat{y}^0(t)-y_\infty| \leq 1$.

(ii) For convenience, we define a function $\Gamma:[0,T]\times (0,1) \to \R$ as
\begin{align}\label{Gamma_def}
\Gamma(t,x):=\EE \left[\tfrac{Y_t^{(0,x)}\left(1-Y_t^{(0,x)}\right)\left( \mu-r-\sigma^2 Y_t^{(0,x)}  \right)}{x(1-x)}  \right]. 
\end{align}
Note that $\Gamma$ does not depend on $\lambda$. The equations \eqref{vx0_exp} and \eqref{vxx0_rep} can be written as
\begin{align}
v_x^0(t,x)=\int_0^{T-t} e^{-\lambda s} \Gamma(s,x)ds, \quad v_{xx}^0(t,x)=\int_0^{T-t} e^{-\lambda s} \Gamma_x(s,x)ds
\end{align}
We can do the similar argument to obtain representations for $v_{xxx}^{0}$ and partial derivatives of $v_x^{0}$ and $v_{xx}^{0}$ with respect to $\lambda$, with the observation that $\Gamma, \Gamma_x, \Gamma_{xx}$ are continuous \& bounded maps on $[0,T]\times (0,1)$. The result is summarized as follows:
\begin{equation}
\begin{split}\label{various_rep}
&\tfrac{\partial}{\partial \lambda} v_x^{0}(t,x)=- \int_0^{T-t} e^{-\lambda s} s \Gamma(s,x)ds, \quad \tfrac{\partial}{\partial \lambda} v_{xx}^{0}(t,x)=-\int_0^{T-t} e^{-\lambda s}s \Gamma_x(s,x)ds \\
& v_{xx}^{0}(t,x)=\int_0^{T-t} e^{-\lambda s} \Gamma_x(s,x)ds,\qquad \quad v_{xxx}^{0}(t,x)=\int_0^{T-t} e^{-\lambda s} \Gamma_{xx}(s,x)ds. 
\end{split}
\end{equation}
Using $Y_t^{(t,x)}=x$, direct computations produce 
\begin{equation}
\begin{split}\label{Gamma_values}
&\Gamma(0,x)=\mu-r-\sigma^2 x, \quad \Gamma_x(0,x) = - \sigma^2, \quad \Gamma_{xx}(0,x)=0.
\end{split}
\end{equation}
 With \eqref{Gamma_values} and \eqref{y0_lambda_infinity}, we apply \lemref{integral_lambda_limit} to \eqref{various_rep} and conclude \eqref{various_limit1}.

(iii) The mean value theorem and $\eqref{vx0_bound}$ produce
\begin{displaymath}
\begin{split}
\hat y^{0,\lambda+\delta}(t) - \hat y^{0,\lambda}(t) = -  \tfrac{ v_x^{0,\lambda+\delta}(t,\hat y^{0,\lambda}(t)) - v_x^{0,\lambda}(t,\hat y^{0,\lambda}(t))}{v_{xx}^{0,\lambda+\delta}(t, z(\lambda,\delta))} \quad \textrm{for  $z(\lambda,\delta)$ between $\hat y^{0,\lambda}(t)$ and $\hat y^{0,\lambda+\delta}(t)$},
\end{split}
\end{displaymath}
where we specify the dependence on $\lambda$ for clarity. 
Since $\frac{\partial}{\partial \lambda} v_x^{0}$ exists (see \eqref{various_rep}), the above equality and \eqref{gxx>0_all} ensure that $\hat y^{0,\lambda}(t)$ is differentiable with respect to $\lambda$ and
\begin{equation}
\begin{split}\label{y_diff_lambda_derive}
 \tfrac{\partial}{\partial \lambda} \hat y^{0}(t) = - \frac{ \frac{\partial}{\partial \lambda} v_x^{0}(t,x)\big|_{x=\hat y^{0}(t)} }{ v_{xx}^{0}(t,\hat y^{0}(t))}.
\end{split}
\end{equation}
Observe that the bounds for $\frac{\partial}{\partial x}Y_t^{(0,x)}$ and $\frac{Y_t^{(0,x)}-x}{x(1-x)}$ in \eqref{YD_bound} and \eqref{yyy_bound} do not depend on the variable $x$. Therefore, the following convergence is uniform on $x \in (0,1)$:
\begin{align}
\Gamma_x(t,x)=-\EE \left[\sigma^2 \left(\tfrac{\partial}{\partial x}Y_t^{(0,x)}\right)^2 + 2\left( \mu-r-\sigma^2 Y_t^{(0,x)}  \right) \left( \tfrac{Y_t^{(0,x)}-x}{x(1-x)}\right) \tfrac{\partial}{\partial x}Y_t^{(0,x)} \right]  \xrightarrow[ t \downarrow 0]{} -\sigma^2. \nonumber
\end{align}
Hence, there is a constant $\tilde T\in (0,T)$ such that $\Gamma_x(t,x)\leq -\frac{\sigma^2}{2}$ for $(t,x)\in [0,\tilde T] \times (0,1)$. This observation and the expression of $v_{xx}$ in \eqref{various_rep} imply
\begin{align}\label{lambda_vxx_bound2}
\lambda v_{xx}(t,\hat{y}^0(t)) \leq -\tfrac{\sigma^2}{2} \left(1-e^{-\lambda(T-t)} \right)  \quad \textrm{for} \quad (t,x)\in [T-\tilde T, T) \times (0,1).
\end{align}
We obtain $\| \Gamma_{xt} \|_\infty<\infty$ by using Ito's formula with \eqref{Y_SDE} and the bounds \eqref{YD_bound} and \eqref{yyy_bound}. Then, \eqref{various_rep} and \eqref{Gamma_values}  imply
\begin{align}
\left| \lambda v_{xx}(t,\hat{y}^0 (t)) + \sigma^2 \right| &= \left| \lambda \int_0^{T-t} e^{-\lambda s} s \, \left( \tfrac{\Gamma_x(s,\hat{y}^0(t)) -\Gamma_x(0,\hat{y}^0(t))}{s}    \right) ds + \sigma^2 e^{-\lambda(T-t)} \right|   \nonumber\\
&\leq \| \Gamma_{xt} \|_\infty \left( \tfrac{1-e^{-\lambda(T-t)}}{\lambda} -(T-t) e^{-\lambda(T-t)} \right) + \sigma^2 e^{-\lambda(T-t)}. \nonumber
\end{align}
This implies that there exists a constant $\tilde \Lambda$ (may depend on $\tilde T$) such that
\begin{align}\label{lambda_vxx_bound}
 \lambda v_{xx}(t,\hat{y}^0 (t)) \leq -\tfrac{\sigma^2}{2} \quad \textrm{for} \quad (t,\lambda)\in [0,T-\tilde T] \times [\tilde \Lambda,\infty).
 \end{align}
Using \eqref{various_rep} and \eqref{Gamma_values}, we obtain
\begin{align}
&\left| \lambda^3 \tfrac{\partial}{\partial \lambda} v_x^{0}(t,x)\big|_{x=\hat y^{0}(t)} \right| \nonumber\\
&\leq \left|  \lambda^3 \int_0^{T-t} e^{-\lambda s} s^2 \left(\tfrac{\Gamma(s, \hat{y}^0(t))-\Gamma(0, \hat{y}^0(t))}{s}\right) ds \right| + \left| \sigma^2(y_\infty- \hat{y}^0(t))\lambda^3 \int_0^{T-t} e^{-\lambda s} s \,ds \right| \nonumber\\
&\leq C\left( 1-e^{-\lambda(T-t)} + \lambda(T-t)\big(1+\lambda(T-t) \big) e^{-\lambda(T-t)} \right) \quad \textrm{for} \quad (t,\lambda)\in [0,T]\times[1,\infty), \label{lambda3_vx_bound}
 \end{align}
 where the second inequality is due to $\| \Gamma_t \|_\infty<\infty$ and \eqref{y0_lambda_infinity}.
 
 From \eqref{y_diff_lambda_derive}, we obtain the boundedness of $\left|\lambda^2 \tfrac{\partial}{\partial \lambda} \hat y^{0}(t) \right|$: 
 \begin{align}
&\sup_{(t,\lambda)\in [0,T-\tilde T]\times [\tilde \Lambda, \infty)} \left|\lambda^2 \tfrac{\partial}{\partial \lambda} \hat y^{0}(t) \right| <\infty \quad \textrm{due to \eqref{lambda_vxx_bound}, \eqref{lambda3_vx_bound}, $\sup_{x>0} x(1+x) e^{-x}<\infty$,} \nonumber\\
&\sup_{(t,\lambda)\in [0,T-\tilde T]\times [1, \tilde \Lambda]} \left|\lambda^2 \tfrac{\partial}{\partial \lambda} \hat y^{0}(t) \right| <\infty \quad \textrm{due to the continuity on compact set},\nonumber\\
&\sup_{(t,\lambda)\in [T-\tilde T,T)\times [1,\infty)} \left|\lambda^2 \tfrac{\partial}{\partial \lambda} \hat y^{0}(t) \right| <\infty \quad \textrm{due to \eqref{lambda_vxx_bound2}, \eqref{lambda3_vx_bound}, $\sup_{x>0}\tfrac{ x(1+x) e^{-x}}{1-e^{-x}}<\infty$.}\nonumber
\end{align}
Therefore, we conclude \eqref{y_lambda_derivative_bound}.

(iv) The bounds \eqref{YD_bound} and \eqref{y_lambda_derivative_bound} enable us to use Leibniz integral rule to obtain
\begin{align}
\tfrac{\partial}{\partial \lambda}  \EE \left[ \left\vert Y_{s}^{(t, \hat{y}^{0}(t))} - \hat{y}^{0}(s) \right\vert \right] 
= \EE \left[ \left( \tfrac{\partial}{\partial \lambda}\hat{y}^0(t) \cdot \tfrac{\partial}{\partial x} Y_s^{(t,x)} \Big|_{x=\hat{y}^0(t)}- \tfrac{\partial}{\partial \lambda}\hat{y}^0(s) \right)  \cdot \sgn\left( Y_{s}^{(t, \hat{y}^{0}(t))} - \hat{y}^{0}(s) \right)   \right]. \nonumber
\end{align}
The above expression, together with the bounds \eqref{YD_bound} and \eqref{y_lambda_derivative_bound}, implies \eqref{g_lambda2_bound}.

(v) Explicit computations using the expression of $Y_s^{(t,x)}$ in \eqref{YL_def} produce
\begin{align}
\tfrac{\EE \left[ \left| Y_s^{(t, x)} - x \right| \right]}{\sqrt{s - t}}&= \int_\R \tfrac{x(1-x)}{x+(1-x) e^{(r-\mu+\frac{\sigma^2}{2})u^2 - \sigma u z}} \cdot  \tfrac{\big| 1-e^{(r-\mu+\frac{\sigma^2}{2})u^2 - \sigma u z}\big|}{u} \cdot \tfrac{e^{-\frac{z^2}{2}}}{\sqrt{2\pi}} dz \bigg|_{u=\sqrt{s-t}} \label{sqrt_g_exp}\\
&\leq \int_\R  \frac{e^{|r-\mu+\frac{\sigma^2}{2}| u^2 + \sigma u |z|} -e^{-|r-\mu+\frac{\sigma^2}{2}| u^2 - \sigma u |z|} }{u}   \cdot \tfrac{e^{-\frac{z^2}{2}}}{\sqrt{2\pi}} dz \bigg|_{u=\sqrt{s-t}} \nonumber\\
&= \frac{e^{(|r-\mu+\frac{\sigma^2}{2}|+\frac{\sigma^2}{2})u^2 }\left(1+\int_0^{\frac{\sigma u}{\sqrt{2}}}\frac{2 e^{-z^2}}{\sqrt{\pi}}dz \right)-e^{(-|r-\mu+\frac{\sigma^2}{2}|+\frac{\sigma^2}{2})u^2 }\left(1-\int_0^{\frac{\sigma u}{\sqrt{2}}}\frac{2 e^{-z^2}}{\sqrt{\pi}}dz \right)}{u} \Bigg|_{u=\sqrt{s-t}},  \nonumber
\end{align}
where we use $|1-e^a|\leq e^{|a|}-e^{-|a|}$ for $a\in \R$ to obtain the inequality.
The last expression converges to $\tfrac{2 \sqrt{2} \sigma}{\sqrt{\pi}}$ as $u \downarrow 0$, therefore, it is bounded on $u\in (0,\sqrt{T}]$. Then, we conclude \eqref{sqrt_g_bound}.

Observe that the integrand in \eqref{sqrt_g_exp} is bounded by $e^{C |z|} \cdot \tfrac{e^{-\frac{z^2}{2}}}{\sqrt{2\pi}}$ for a constant $C$ independent of $(u,x,z)\in (0,\sqrt{T}]\times (0,1)\times \R$. Since this bound is integrable with respect to $z$, we apply the dominated convergence theorem to \eqref{sqrt_g_exp} and obtain 
\begin{align}
  \lim_{s \downarrow t} \tfrac{\EE \left[ \left| Y_s^{(t, x)} - x \right| \right]}{\sqrt{s - t}} 
  &=   \int_\R x(1-x)\sigma |z|\cdot \tfrac{e^{-\frac{z^2}{2}}}{\sqrt{2\pi}} dz  = \sigma \sqrt{\tfrac{2}{\pi}}\, x(1-x). \nonumber
\end{align}
Therefore, we conclude \eqref{sqrt_g_lim}.

\section{Supplementary Lemmas}
\begin{lemma}\label{meas_lem}
Let $F:[0,T]\times [0,1]^2\to \R$ be a continuous function. We define $f:[0,T]\times[0,1]\to [0,1]$ as
\begin{align}
f(t,x):=\max \left\{ z: z\in \argmax_{y\in [0,1]} F(t,x,y)  \right\},
\end{align}  
then $f$ is upper semicontinuous (which is obviously Borel-measurable). 
\end{lemma}
\begin{proof}
This type of result is well-known (e.g., see p. 153 in \cite{bertsekas2004stochastic}), but we give a short proof here for the sake of self-containedness. 

Since $F$ is continuous, $ \argmax_{y\in [0,1]} F(t,x,y) $ is a nonempty closed subset of $[0,1]$, so the maximum element of $ \argmax_{y\in [0,1]} F(t,x,y) $ exists and $f$ is well-defined. Let $\{(t_n,x_n)\}_{n\in \N}\subset [0,T]\times [0,1]$ be a sequence converging to $(t_\infty,x_\infty)$ such that $\lim_{n\to \infty} f(t_n,x_n)$ exists. Then, by definition of $f$, 
\begin{align}
F(t_n,x_n,f(t_\infty,x_\infty))\leq F(t_n,x_n,f(t_n,x_n)).
\end{align}
We let $n\to \infty$ above and using the continuity of $F$ to obtain
\begin{align}
F(t_\infty,x_\infty,f(t_\infty,x_\infty)) \leq F(t_\infty,x_\infty,\lim_{n\to \infty} f(t_n,x_n)).
\end{align}
This implies that $\lim_{n\to \infty} f(t_n,x_n) \in  \argmax_{y\in [0,1]} F(t_\infty,x_\infty,y)$, and the definition of $f$ ensures $$f(t_\infty,x_\infty)\geq  \lim_{n\to \infty} f(t_n,x_n).$$
Therefore, $f$ is upper semicontinuous. 
\end{proof}

\begin{lemma}\label{integral_lambda_limit}
Let $f(s,x) : [0,t] \times (0,1) \rightarrow \R$ be a continuous function, and $g(\lambda): [1,\infty) \rightarrow (0,1)$ be a function satisfying $\lim_{\lambda \rightarrow \infty} g(\lambda)=x_\infty \in (0,1)$. Then, for $\alpha\in \{0,1,2,3,4\}$ and $t > 0$,
\begin{equation}
\begin{split}\label{lemma2_result}
&\lim_{\lambda \to \infty} \lambda^{\frac{\alpha}{2}+1} \int_0^{t} e^{-\lambda s}s^{\frac{\alpha}{2}} f(s,g(\lambda)) ds = c_\alpha \cdot f(0,x_\infty), 
\end{split}
\end{equation}
where $c_0=1$, $c_1=  \frac{\sqrt{\pi}}{2}$, $c_2=  1$, $c_3=  \frac{3\sqrt{\pi}}{4}$, $c_4=  2$. Also, there exists a constant $C$ such that
\begin{equation}\label{lemma2_result2}
 \lambda^{\frac{\alpha}{2}+1} \int_0^t e^{-\lambda s} s^{\frac{\alpha}{2}} ds \leq  C, \quad \textrm{for} \quad (t,\lambda,\alpha)\in [0,\infty)\times [1,\infty)\times \{0,1,2,3,4\}.
\end{equation}
\end{lemma}
\begin{proof}
Let $\eta>0$ be a given constant. The uniform continuity of $f$ on a compact set containing the point $(0,x_\infty)$, together with $\lim_{\lambda \rightarrow \infty} g(\lambda)=x_\infty \in (0,1)$, implies that there exists $\delta>0$ such that
\begin{align}\label{e-d}
\left| f(s,g(\lambda)) - f(0,x_\infty)  \right| \leq \eta \quad \textrm{for any} \quad (s,\lambda) \in [0,  \delta] \times \left[ \tfrac{1}{\delta}, \infty \right).
\end{align}
Simple computations and \eqref{e-d} produce
\begin{displaymath}
\begin{split}
\limsup_{\lambda \to \infty} \Big| \lambda \int_0^t e^{-\lambda s }  f(s,g(\lambda))  ds - & f(0,x_\infty) \Big|
= \limsup_{\lambda \to \infty} \bigg| \int_0^{\delta} \lambda e^{-\lambda s} (f(s,g(\lambda))  ds - f(0,x_\infty))ds \\
&  +\int_{ \delta}^t \lambda e^{-\lambda s} (f(s,g(\lambda))  ds - f(0,x_\infty))ds - e^{-\lambda t} f(0,x_\infty)  \bigg|  \leq \eta.
\end{split}
\end{displaymath}  
Since $\eta>0$ can be arbitrary small, we conclude \eqref{lemma2_result} for the case of $\alpha=0$. The other cases in \eqref{lemma2_result} can be obtained by the same way as above, using the following expressions:
\begin{equation}
\begin{split}\label{lambda_integrals}
&\lambda^{\frac{3}{2}} \int_0^t e^{-\lambda s} s^{\frac{1}{2}}   ds = -\sqrt{\lambda t} e^{-\lambda t} + \int_0^{\sqrt{\lambda t}} e^{-s^2} ds \xrightarrow[\lambda\to \infty]{} \frac{\sqrt{\pi}}{2}, \\
&\lambda^2 \int_0^t e^{-\lambda s}s \,ds =1 -(1+\lambda t) e^{-\lambda t}\xrightarrow[\lambda\to \infty]{} 1, \\
&\lambda^{\frac{5}{2}} \int_0^t e^{-\lambda s} s^{\frac{3}{2}}   ds = -\tfrac{(3+2\lambda t)\sqrt{\lambda t} e^{-\lambda t}}{2} + 
\frac{3}{2}\int_0^{\sqrt{\lambda t}} e^{-s^2} ds \xrightarrow[\lambda\to \infty]{} \frac{3\sqrt{\pi}}{4},\\
&\lambda^3 \int_0^t e^{-\lambda s}s^2 ds =2 -(2+\lambda t(2+\lambda t)) e^{-\lambda t}\xrightarrow[\lambda\to \infty]{} 2. \\
\end{split}
\end{equation}  
One can easily observe that $\sup_{x>0} x^{\frac{\alpha}{2}} e^{-x} <\infty$ for $\alpha\in \{0,1,2,3,4\}$. This observation and the explicit expressions in \eqref{lambda_integrals} produce the bound \eqref{lemma2_result2}.
\end{proof}

\end{document}